\documentclass[11pt]{article}
\usepackage{makeidx}
\usepackage{amsmath}
\usepackage{amssymb}
\usepackage{amsfonts}
\usepackage[all]{xy}
\newcommand{\QCB}{\mathsf{QCB}}
\newcommand{\QCBZ}{\mathsf{QCB_0}}
\newcommand{\QN}{\mathsf{QN}}
\newcommand{\Top}{\mathsf{Top}}
\newcommand{\Seq}{\mathsf{Seq}}
\newcommand{\Lim}{\mathsf{Lim}}
\newcommand{\oLim}{\omega\mathsf{Lim}}

\newcommand{\oWLim}{\omega\mathsf{WLim}}
\newcommand{\oPFil}{\omega\mathsf{PFil}}
\newcommand{\Pomega}{{P\hspace*{-1pt}\omega}}
\newcommand{\parto}{\dashrightarrow}  
\newcommand{\Nomega}{\IN^\IN}
\newcommand{\Somega}{\Sigma^\IN}
\newcommand{\sdrep}{\varrho_{\mathrm{sd}}}
\newcommand{\sdrepIZ}{\widetilde{\rho_{\mathrm{sd}}}}
\newcommand{\decrep}{\varrho_{\mathrm{dec}}}
\newcommand{\RepN}{\mathsf{Rep}}
\newcommand{\RepEff}{\mathsf{Rep}_\mathrm{eff}}
\newcommand{\MRep}{\mathsf{MRep}}
\newcommand{\EffQCB}{\mathsf{EffQCB}}
\newcommand{\EffLim}{\mathsf{EffLim}}
\newcommand{\kTop}{{\mathsf{kTop}}}
\newcommand{\Ccg}{{\mathsf{CCG}}}
\newcommand{\Equ}{{\mathsf{Equ}}}
\newcommand{\oEqu}{{\omega\mathsf{Equ}}}
\newcommand{\Mod}{\mathsf{Mod}}
\newcommand{\w}[1]{{\texttt{#1}}}
\newcommand{\To}[1]{\mathrel{\rightarrow_{#1}}}
\newcommand{\seq}{{\mathit{seq}}}
\newcommand{\En}{\mathit{En}}
\newcommand{\leqcp}{\mathrel{\leq_{\mathrm{cp}}}}
\newcommand{\leqt}{\mathrel{\leq_{\mathrm{t}}}}
\newcommand{\equivcp}{\mathrel{\equiv_{\mathrm{cp}}}}
\newcommand{\equivt}{\mathrel{\equiv_{\mathrm{t}}}}
\newcommand{\dom}{\mathit{dom}}
\newcommand{\Cls}{\mathit{Cls}}
\newcommand{\Graph}{\mathit{Graph}}

\newcommand{\Tbigcup}{\mathop{\textstyle\bigcup}}

\newcommand{\Dsum}{\mathop{\displaystyle\sum}}

\newcommand{\mybigboxtimes}{\mathop{\boxtimes}}
\newcommand{\mybigboxplus}{\mathop{\boxplus}}
\newcommand{\mybigwedge}{\mathop{\wedge}}
\newcommand{\etaFww}{{\bar{\eta}}}
\newcommand{\fst}{\mathit{fst}}
\newcommand{\snd}{\mathit{snd}}

\newcommand{\calAm}{{\mathcal{A}_-}}
\newcommand{\calAp}{{\mathcal{A}_+}}
\newcommand{\calB}{\mathcal{B}}

\newcommand{\calO}{\mathcal{O}}
\newcommand{\calT}{\mathcal{T}}
\newcommand{\calV}{\mathcal{V}}
\newcommand{\setA}{\mathit{A}}
\newcommand{\setC}{\mathit{C}}
\newcommand{\setO}{\mathit{O}}
\newcommand{\IL}{\mathbb{L}}

\newcommand{\INu}{{\mathbb{N}_\infty}}
\newcommand{\IN}{\mathbb{N}}
\newcommand{\IQ}{\mathbb{Q}}
\newcommand{\IS}{\mathbb{S}} 
\newcommand{\IR}{\mathbb{R}}
\newcommand{\IZ}{\mathbb{Z}}
\newcommand{\cf}{\mathit{cf}}
\newcommand{\Kup}{{\mathcal{K}}}   
\newcommand{\KViet}{{\mathcal{K}_\mathrm{Viet}}}
\newcommand{\KQ}{\mathsf{KQ}}
\newcommand{\WA}{{\mathsf{A}}} 
\newcommand{\WK}{{\mathsf{K}}}
\newcommand{\WW}{{\mathsf{W}}}
\newcommand{\WX}{{\mathsf{X}}}
\newcommand{\WY}{{\mathsf{Y}}}
\newcommand{\WZ}{{\mathsf{Z}}}
\newcommand{\PP}{{\mathbf{P}}}
\newcommand{\XX}{{\mathbf{X}}}
\newcommand{\YY}{{\mathbf{Y}}}
\newcommand{\ZZ}{{\mathbf{Z}}}
\newtheorem{theorem}{Theorem}[section]
\newtheorem{proposition}[theorem]{Proposition}

\newtheorem{definition}[theorem]{Definition}
\newtheorem{example}[theorem]{Example}

\newenvironment{proof}{\pagebreak[3]\noindent\emph{Proof.}}
{\par\hspace*{\fill}$\square$\pagebreak[3]\medskip}
\newenvironment{proofof}[1]{\pagebreak[3]\noindent\emph{Proof of #1.}}
{\pagebreak[3]\medskip}
\makeindex                     

\begin{document}

\title{Admissibly Represented Spaces and Qcb-Spaces\footnote{This survey is accepted as a chapter in the \emph{Handbook of Computability and Complexity in \newline Analysis} \cite{BH:handbook},
to be published by the Springer-Verlag (www.springer.com).}}
\author{Matthias Schr{\"o}der}
\date{\quad}

\maketitle

\abstract{%
A basic concept of Type Two Theory of Effectivity (TTE) is the notion of an admissibly represented space.
Admissibly represented spaces are closely related to qcb-spaces.
The latter form a well-behaved subclass of topological spaces.
We give a survey of basic facts about Type Two Theory of Effectivity, 
admissibly represented spaces, qcb-spaces and effective qcb-spaces.
Moreover, we discuss the relationship of qcb-spaces to other categories relevant to Computable Analysis.
}



\section{Introduction}
\label{sec:intro}
Computable Analysis investigates computability on real numbers and related spaces.
Type Two Theory of Effectivity (TTE) constitutes a popular approach to Computable Analysis,
providing a rigorous computational framework for non-discrete spaces with cardinality of the continuum (cf.\ \cite{Wei87,Wei00}). 
The basic tool of this framework are representations. 
A representation equips the objects of a given space with names,
giving rise to the concept of a represented space.
Computable functions between represented spaces are those which are realized by a computable function on the names.
The ensuing category of represented spaces and computable functions enjoys excellent closure properties.

Any represented space is equipped with a natural topology, turning it into a qcb-space.
Qcb-spaces form a subclass of topological spaces with a remarkably rich structure.
For example it is cartesian closed, hence products and function spaces can be formed.

Admissibility is a notion of topological well-behavedness for representations.
The category of admissibly represented spaces and continuously realizable functions
is equivalent to the category $\QCBZ$ of qcb-spaces with the $T_0$-property.
In particular every qcb$_0$-space can be equipped with an admissible representation. 
Qcb$_0$-spaces form exactly the class of topological spaces which can be handled appropriately by the framework of Type Two Theory of Effectivity.

An important subcategory of admissibly represented spaces is the class of effective qcb-spaces. 
These are qcb-spaces endowed with representations that are admissible in a computable sense.
Computable metric spaces endowed with their Cauchy representations yield important examples.
The category $\EffQCB$ of effective qcb-spaces and computable functions has nice closure properties. 
Therefore Type Two Theory of Effectivity is applicable to a large class of important spaces occurring in Analysis.

This paper is organised as follows.
In Section~\ref{sec:Repspaces} we recall basic notions of Type Two Theory of Effectivity
with emphasis on represented spaces and their closure properties.
Section~\ref{sec:admissibility} is devoted to the notion of an admissible representation for general topological spaces.
In Section~\ref{sec:QCB} we present and discuss qcb-spaces.
Furthermore we summarise basic effectivity properties of effective qcb-spaces.
Finally we compare in Section~\ref{sec:relationship:QCB} the category $\QCB$ of qcb-spaces with other categories relevant to Computable Analysis.


\pagebreak[3]
\section{Represented Spaces}
\label{sec:Repspaces}

We present basic concepts of 
Type Two Theory of Effectivity\index{Type Two Theory of Effectivity}(TTE),
as developed by K.~Weihrauch (\cite{Wei87,Wei00}).
TTE offers a model for rigorous computation for uncountable spaces occurring in Analysis like the real numbers, separable Banach spaces and Silva spaces.
The key idea is the notion of represented spaces and of computable functions between them.
We collect important properties of the ensuing categories $\RepN$ and $\RepEff$ of represented spaces.
Finally we mention the notion of multirepresentations generalising ordinary representations.


\subsection{Representations}
\label{sub:Representations}

In classical recursion theory, computability on a countable space is defined by lifting the well-established notion of a computable natural number function via an appropriate numbering of $X$.
A \emph{numbering} $\nu$ of a set $X$ is a partial surjection $\nu\colon \IN \parto X$.
Partiality means that not every number is a representative of an element in $X$.
For uncountable sets $X$ this approach fails, because the large cardinality excludes the existence of a numbering.
The basic idea of the Type Two Theory of Effectivity (TTE) is to encode the objects of $X$ by infinite words over a finite or countably infinite alphabet $\Sigma$,
i.e.\ by elements of the set $\Somega=\{ p\,|\, p\colon \IN \to \Sigma\}$, and to compute on these names.
The corresponding partial surjection $\delta \colon \Somega \parto X$ mapping every name $p \in \dom(\delta)$ to the encoded element $\delta(p)$ is called a \emph{representation}\index{representation} of $X$.
If $\delta(p)=x$ then $p$ is called a \emph{name} or \emph{representative} of the element $x \in X$.
 
\pagebreak[3]
\begin{example}[The decimal representation] \rm
 \label{ex:decrep}
  A well-known representation of the set $\IR$ of real numbers is the \emph{decimal representation} $\decrep$.
  As its alphabet $\Sigma_{\mathrm{dec}}$ it uses $\{\w{0},\dotsc,\w{9},\w{-},\w{.}\}$.
  It is defined by
  \begin{alignat*}{2}
   &\decrep(a_{-k} \dotsc a_0 \w{.} a_1 a_2 \dotsc)
   &&:= +\Dsum\limits_{i=-k}^\infty \dfrac{ a_i}{10^{i}} \,,
 \\ 
   &\decrep(\w{-}
       a_{-k} \ldots a_0 \texttt{.} a_1 a_2 \dotsc)
   &&:= -\Dsum\limits_{i=-k}^\infty \dfrac{a_i}{10^{i}}
  \end{alignat*}
  for all $k \in \IN$ and all digits
  $a_{-k},a_{-k+1},\dotsc \in \{ \w{0},\dotsc,\w{9} \}$.  
\end{example}
\par
For establishing a simple framework of computability for general spaces
it is reasonable to restrict oneself to a single, countably infinite alphabet.
We choose the set $\IN$ of natural numbers as our alphabet.
So we use the \emph{Baire space} $\Nomega$ as our space of representatives. 
One might also use other spaces of representatives.
For example, representations of real numbers by elements of $\IZ^\IN$
were used by J.~Hauck in \cite{Hau73} and implicitly by A.~Grzegorczyk in \cite{Grz57}.
Moreover, J.~Hauck employed representations based on $\IZ^\IN$ to introduce computability
for functions on topological spaces with a recursive basis in \cite{Hau81}.
M.~Escard{\'o} used domains for his approach to higher-order exact real number computation \cite{Esc96}.
%
%
\subsection{The Baire Space}
\label{sub:Bairespace}

The Baire space\index{Baire space} $\Nomega$ 
is a topological space that has all sequences of natural numbers as its carrier set.
Its topology is called \emph{Baire topology} and is generated by the basic opens
\[
 u\Nomega:=\{ p \in \Nomega \,|\, \text{$u$ is a prefix of $p$}\}\,,
\]
where $u$ ranges over the set $\IN^*$ of all finite strings over $\IN$.
This base is countable, 
hence $\Nomega$ is a countably-based (= second-countable) topological space.
Any open set has the form
\[
  W\Nomega:=\{ p \in \Nomega \,|\, \text{$p$ has a prefix in $W$}\}\,,
\]
where $W$ is an arbitrary subset of $\IN^\ast$. 
The Baire space is separable: the elements $u\w{0}^\omega$, 
where $u\w{0}^\omega$ denotes the sequence with prefix $u \in \IN^*$ followed by infinitely many $\w{0}$'s, form a countable dense subset.

The basic open sets $u\Nomega$ are not only open, but also closed.
Such sets are usually called \emph{clopen} (abbreviating closed-and-open).
Topological spaces having a basis consisting of clopen sets are called \emph{zero-dimensional}.

The Baire topology is induced by the metric $d_{\Nomega}\colon \Nomega \times \Nomega \to \IR_{\geq 0}$ defined by
\[
  d_{\Nomega}(p,p)=0 
  \;\;\text{and}\;\;
  d_{\Nomega}(p,q):=2^{-\textstyle\min\{i \in \IN \,|\, p(i) \neq q(i)\}}
\]
for all $p \neq q$ in $\Nomega$.
Note that topological continuity coincides with $\epsilon$-$\delta$-continuity for functions between metrisable spaces.
The metric $d_{\Nomega}$ is complete, meaning that every Cauchy sequence $(p_n)_n$
has a limit $p:=\lim\nolimits_{n \to \infty} p_n$.
Therefore $\Nomega$ is a separable completely metrisable topological space.
Such spaces are referred to as \index{Polish space}\emph{Polish spaces}. 

The Baire space $\Nomega$ is homeomorphic to its own binary product $\Nomega \times \Nomega$ 
and to the countable product $\prod_{i \in \IN} \Nomega$ of itself. 
A pairing function
$\langle \cdot,\cdot \rangle\colon \Nomega \times \Nomega \to \Nomega$ 
for $\Nomega$ is given by
\[
  \langle p,q\rangle(2n)=p(n) \;\;\text{and}\;\; \langle p,q\rangle(2n+1):=q(n)
  \,.
\]
We denote the projections of its inverse by $\pi_1,\pi_2\colon \Nomega \to \Nomega$.
Moreover, a countable tupling function
$\langle \cdot \rangle\colon \prod_{i\in\IN}\Nomega \to \Nomega$ is defined by
\[
 \langle p_0,p_1,p_2,\dotsc \rangle(k):=p_{\fst(k)}(\snd(k)),
\]
where $\fst,\snd\colon \IN \to \IN$ are the computable projections of a canonical computable pairing function on $\IN$.
We denote by $\pi_{\infty,i} \colon \Nomega \to \Nomega$ the $i$-th projection of the inverse.
The pairing function $\langle \cdot,\cdot \rangle$ on $\Nomega$ as well as the projection functions $\pi_1,\pi_2, \pi_{\infty,0}, \pi_{\infty,1},\dotsc $ are computable in the sense of Subsection~\ref{sub:ComputabilityBaireSpace}. 
%
%
\subsection{Computability on the Baire Space}
\label{sub:ComputabilityBaireSpace}

Computability for functions on the Baire space $\Nomega$ can be defined either by computable monotone string functions 
or by Type-2 machines. 
We present the first option.

A string function $h\colon \IN^* \to \IN^*$ is called \emph{monotone},
if it is monotone with respect to the prefix order $\sqsubseteq$
on the set $\IN^\infty:=\IN^* \cup \Nomega$ of finite and countably infinite strings over $\IN$.
A monotone string function $h\colon \IN^* \to \IN^*$ generates
a partial function $h^\omega$ on the Baire space defined by
\begin{align*}
  &\dom(h^\omega):=\big\{ p \in \Nomega \,\big|\,
     \{ h(u) \,|\, \text{$u$ is a prefix of $p$}\} \text{ is infinite} \big\}
  \\   
  &h^\omega(p):= \sup\nolimits_\sqsubseteq
   \big\{ h(u) \,\big|\, \text{$u$ is a prefix of $p$} \big\}
\end{align*}
for all $p \in \dom(h^\omega)$.
For example, the identity function on $\IN^\ast$ generates the identity function on $\Nomega$.

A partial function $g\colon \Nomega \parto \Nomega$ is called \emph{computable},
if there is a computable monotone function $h\colon \IN^* \to \IN^*$ generating $g$, i.e. $g=h^\omega$. 
Computability of multivariate functions on $\Nomega$ is defined analogously.
A \emph{computable element} of $\Nomega$ is just some $p \in \Nomega$ which is computable as a natural number function in sense of discrete computability theory.

\pagebreak[3]
\begin{example} \rm
 The pairing function $\langle \cdot,\cdot \rangle$ and the projection functions $\pi_1,\pi_2,\pi_{\infty,i}$ 
 from Subsection~\ref{sub:Bairespace} are computable.
\end{example}

Computable functions on $\Nomega$ are closed under composition.

\pagebreak[3]
\begin{proposition}[Composition preserves computability] \quad
\label{p:composition:computable}
\begin{enumerate}\vspace*{-0.5ex}
 \item 
   Let $f \colon (\Nomega)^k \parto \Nomega$ and $g_1,\dotsc,g_k \colon (\Nomega)^\ell \parto \Nomega$
   be computable functions.
   Then the composition $f \circ (g_1 \times\dotsc\times g_k)$ is computable.
 \item
   Let $g$ be a partial computable function on $\Nomega$ and let $p \in \dom(g)$ be a computable element of $\Nomega$.
   Then $g(p)$ is a computable element of $\Nomega$.
\end{enumerate}   
\end{proposition}
The following key observation links computability theory with topology:
any computable function on the Baire space is continuous.

\pagebreak[3]
\begin{proposition}[Computability implies continuity] 
\label{p:ComputableImpliesContinuous}
 \quad\\ \noindent
 Any computable function $g \colon \Nomega \times\dotsc\times \Nomega \parto \Nomega$ is topologically continuous.
\end{proposition}

\begin{proof}
 We restrict ourselves to the unary case.
 So there is a computable monotone function $h\colon \IN^* \to \IN^*$ such that $h^\omega=g$.
 Let $w\Nomega$ be a basic open (see Subsection~\ref{sub:Bairespace})
 and let $p \in \dom(g)$ with $p \in w\Nomega$.
 Then there is some $u \in \IN^*$ such that $u \sqsubseteq p$ and $w  \sqsubseteq h(u) \sqsubseteq g(p)$.
 This implies that for all $q \in \dom(g)$ with $u \sqsubseteq q$ we have $g(q) \in w\Nomega$.
 Therefore $p \in u\Nomega \cap \dom(g) \subseteq g^{-1}[w\Nomega]$.
 We conclude that $g$ is topologically continuous.
\end{proof}

\pagebreak[3]
\begin{example}[An incomputable function on $\Nomega$]\rm \quad\\ \noindent 
The test function $\mathit{is\_zero}\colon\Nomega \to \Nomega$ defined by
\[
 \mathit{is\_zero}(p):= \left\{
 \begin{array}{ll}
   \w{0}^\omega & \text{if $p=\w{0}^\omega$}
   \\
   \w{1}^\omega & \text{otherwise}
 \end{array}\right.
\]
is incomputable by being discontinuous.
\end{example}

The set $\mathrm{F}^{\omega\omega}$ of all partial continuous functions on $\IN^\IN$ with a $G_\delta$-domain
can be endowed with an ``effective'' representation $\etaFww$  
which satisfies the computable utm-Theorem, the computable smn-Theorem, and the continuous smn-Theorem.

\pagebreak[3]
\begin{proposition}
\label{p:etaFww}
 There is a representation $\etaFww\colon \Nomega \to \mathrm{F}^{\omega\omega}$ such that:
\begin{enumerate}
  \item 
   The evaluation function $(p,q) \mapsto \etaFww(p)(q)$ is computable.
  \item 
   For every computable function
   $g: \Nomega \times \Nomega \parto \Nomega$
   there is a computable function $s: \Nomega \to \Nomega$ satisfying
   $\etaFww(s(p))(q)= g(p,q)$ for all $p,q \in \Nomega$.
  \item  
   For every partial continuous function
   $g\colon \Nomega \times \Nomega \parto \Nomega$ 
   there is a continuous function $s\colon \Nomega \to \Nomega$
   satisfying $\etaFww(s(p))(q)=g(p,q)$ for all $(p,q) \in \dom(g)$.
\end{enumerate} 
\end{proposition}
Explicit constructions of $\etaFww$ can be found e.g.\
in \cite{Bau01,LN15,Sch:phd,Wei87}. 
The application $p \odot q:= \etaFww(p)(q)$ induces on $\Nomega$ the structure of a partial combinatory algebra (pca), cf.\ \cite{Bau00,LN15}.
The pair $(\Nomega, \odot)$ is often referred to as Kleene's second 
algebra $K_2$.

%
%

\subsection{Represented Spaces}
\label{sub:RepresentedSpaces}

A \emph{represented space}\index{represented space} is a pair $\XX=(X,\delta)$
such that $X$ is a set and $\delta$ is a Baire space representation of $X$,
i.e., a partial surjection from $\Nomega$ onto $X$.
The name `represented space' was coined by V.~Brattka in \cite{Bra96}.

%
%

\subsection{Computable Elements of Represented Spaces}
\label{sub:ComputableElements}

An element $x$ of a represented space $\XX=(X,\delta)$ is called \emph{computable}\index{computable element} 
(or \emph{$\delta$-com\-putable}),
if it has a computable name $p \in \Nomega$.
Similarly, a \emph{computable sequence}\index{computable sequence} $(x_i)_i$ in $\XX$ is a sequence of points such that there is exists a computable $p \in \Nomega$ such that $\delta(\pi^\infty_i(p))=x_i$ for all $i \in \IN$.
The reader should be warned that a sequence of computable elements need not be a computable sequence of elements.

\pagebreak[3]
\begin{example}[Computable real numbers] \rm \quad\\ \noindent
 The real numbers which are $\decrep$-computable are exactly those numbers $x \in \IR$ for which there is a computable sequence $(z_k)_k$ of integers  satisfying $|\tfrac{z_k}{k+1} - x| \leq \tfrac{1}{k+1}$ for all $k \in \IN$.
 The latter is the common definition of a computable real number as introduced by A.~Grzegorczyk (see \cite{Grz57}).
 Rational numbers and $\sqrt{2},\pi,\exp(1)$ are computable real numbers, whereas
 for any non-decidable subset $M$ of natural numbers, $\sum_{i \in M} 4^{-i}$ is an incomputable real number.
\end{example} 

\pagebreak[3]
\begin{example}[Computable sequences of real numbers] \rm \quad\\ \noindent
 A sequence $(x_i)_i$ of real numbers is called \emph{computable}\index{computable sequence} if, and only if, there is a computable double sequence $(q_{i,k})_{i,k}$ of rational numbers such that $|q_{i,k} - x_i| \leq 2^{-k}$ for all $i,k \in \IN$.
 This is equivalent to saying that $(x_i)_i$ is a computable sequence of the represented space $(\IR,\decrep)$. 
\end{example}

\subsection{Computable Realizability}
\label{sub:ComputableRealizability}

A function $f$ between two represented spaces $\XX=(X,\delta)$ and $\YY=(Y,\gamma)$ is called \emph{computable}\index{computable function},
if there is a partial computable function $g\colon \Nomega \parto \Nomega$
\emph{realizing $f$}, meaning that $f\delta(p)=\gamma g(p)$ holds for all $p \in \dom(f\delta)$.
This notion of computability is also called \emph{$(\delta,\gamma)$-computability} 
or \emph{computable realizability}\index{computable realizability} w.r.t.\ $\delta$ and $\gamma$.

That fact that a \emph{realizer}\index{realizer} $g$ realizes a function $f$ is usually visualised by a commuting diagram like the following:
\begin{equation*} 
  \xymatrix{
    X
    \ar[rr]^{\textstyle f}
    \ar@{}[drr]|{\textstyle\pmb{\circlearrowleft}}
    & &
    Y
   \\
    \Nomega
    \ar@{->>}[u]^{\textstyle\delta}
    \ar[rr]_{\textstyle g}
    & &
     \Nomega
     \ar@{->>}[u]_{\textstyle\gamma} 
    } 
\end{equation*}
We remark that this diagram is slightly imprecise, since it does not capture the situation for $p \notin \dom(f\delta)$: 
in particular $g$ is not required to diverge on names $p \in \dom(\delta) \setminus \dom(f\delta)$.
So a more precise commuting diagram looks like:
\begin{equation*} 
  \xymatrix{
    X
    \ar[rr]^{\textstyle f}
    \ar@{}[drr]|{\textstyle\pmb{\circlearrowleft}}
    & &
    Y
   \\
    \dom(f\delta)
    \ar@{->>}[u]^{\textstyle\delta}
    \ar[rr]_{\textstyle g}
    & &
     \Nomega 
     \ar@{->>}[u]_{\textstyle\gamma} 
    } 
\end{equation*}

\noindent
Computability of partial multivariate functions $f\colon \XX_1 \times\dotsc\times \XX_k \parto \YY$ between represented spaces is defined via the existence of a computable realizer $g\colon (\Nomega)^k \parto \Nomega$.

By Proposition~\ref{p:ComputableImpliesContinuous}, 
any computable function between represented spaces maps computable elements in the domain to computable elements in the image;
moreover composition of functions between represented spaces preserves computability.
Clearly the identity function on a represented space is computable.
So the class of represented spaces and the total computable functions between them
forms a category which we denote by $\RepEff$.

\pagebreak[3]
\begin{example}[Computability with respect to the decimal representation]\label{ex:MultiByThree} \rm 
\noindent
 Real multiplication by $2$ is computable with respect to the decimal representation $\decrep$.
 By contrast, multiplication by $3$ is \emph{not} computable w.r.t.\ $\decrep$.
 This was already observed by A.~Turing in \cite{Tur38}.
 We present the simple well-known proof.
 To fit into our framework, we identify the symbols $\w{-}$ and $\w{.}$ with the numbers ten and eleven, respectively, so that we can view $\Sigma_{\mathrm{dec}}$ as a subset of $\IN$
 and $\decrep$ as a Baire space representation.
 
 Suppose for contradiction that $h$ is a monotone computable function such that $h^\omega$ realizes
 multiplication by $3$ with respect to $\decrep$:
 \begin{equation*} 
  \xymatrix{
    \IR
    \ar[rr]^{\textstyle x \mapsto 3x}
    \ar@{}[drr]|{\textstyle\pmb{\circlearrowleft}}
    & &
    \IR
   \\
    \{\w{0},\dotsc,\w{9},\w{-},\w{.}\}^\IN
    \ar@{->>}[u]^{\decrep}
    \ar[rr]_{\textstyle h^\omega}
    & &
     \{\w{0},\dotsc,\w{9},\w{-},\w{.}\}^\IN
     \ar@{->>}[u]_{\decrep} 
    } 
\end{equation*}
The sequence $p:=\w{0.}\w{3333}\dotsc$ is a name of $\tfrac{1}{3}$.
Hence $h^\omega(p)$ must be equal to $\w{0.999}\dotsc$ or to $\w{1.000}\dotsc$;
let us assume the first case.
By the continuity of~$h^\omega$, there is a finite prefix $u=\w{0.}\w{33}\dotsc\w{3}$ of $p$ such that
$\w{0.}\w{33}\dotsc\w{3} \sqsubseteq q$ implies $\w{0.9} \sqsubseteq h^\omega(q)$.
But then $h^\omega$ works incorrectly on the name $p':=u\w{4}\w{000}\dotsc=\w{0.}\w{33}\dotsc\w{3} \w{4}\w{000}\dotsc$ 
which represents a real number strictly larger than $\tfrac{1}{3}$.
In the other case we obtain a similar contradiction.
We conclude that multiplication by $3$ cannot be computable on $(\IR,\decrep)$.
\end{example}

We obtain that the basic arithmetic operations $+,-,\ast,\div$ on the reals are incomputable with respect to the decimal representation.

%
%

\subsection{The Cauchy Representation of the Real Numbers}
\label{sub:CauchyRepReal}

We present a representation for the real numbers which induces a reasonable notion of computability.
It is known as the \emph{Cauchy representation}\index{Cauchy representation} for the reals.
The Cauchy representation $\varrho_\IR$ is defined by 
\[
  \varrho_\IR(p)=x \;:\Longleftrightarrow\; \forall k\in \IN. |\nu_\IQ(p(k)) -x | \leq 2^{-k},
\]
where $\nu_\IQ$ is 
a canonical numbering\footnote{A canonical numbering of $\IQ$ can defined by  $\nu_\IQ(n):=\tfrac{\nu_\IZ(\fst(n))}{1+\snd(n)}$, where $\fst,\snd$ are the computable projections of a computable pairing function for $\IN$ and $\nu_\IZ$ is a canonical numbering of the integers like the one in Subsection~\ref{sub:SignedDigitRep}.}
of the rational numbers $\IQ$.
So a name of a real number is essentially a rapidly converging sequence of rationals.

We present an incomplete list of functions which are computable with respect to the Cauchy representation.

\pagebreak[3]
\begin{example}[Computable real functions]
\label{ex:comp:realfunctions}\rm \quad\\ \noindent
The following (total or partial) real functions are computable with respect to the Cauchy representation on their natural domain.
\begin{enumerate}
  \item
    $+$, $-$, $\ast$, $\div$
  \item
   $\sqrt{x}$,
   $\sqrt[y]{x}$,
   $\;\exp$, $\log$
  \item 
   $\;\sin$, $\cos$, $\tan$,
   $\;\arcsin$, $\arccos$, $\arctan$
  \item
   $x \mapsto \int_0^x f dt$  for any computable function $f\colon (\IR,\varrho_\IR) \to (\IR,\varrho_\IR)$.
\end{enumerate}
\end{example}

These facts give rise to choose the following as the notion of computability on the real numbers.

\pagebreak[3]
\begin{definition} \rm
 A total or partial real function $f\colon \IR^k \parto \IR$ is called \emph{computable}\index{computable function},
 if it is computable with respect to the Cauchy representation $\varrho_\IR$.
 A real number is \emph{computable}\index{computable element}, if it has a computable name $p$ with $\varrho_\IR(p)=x$.
\end{definition}

This notion of computability is essentially the same as the one considered by A.~Grzegorczyk in \cite{Grz57},
at least for functions with well-behaved domain.
We mention some implementations of exact real arithmetic:
N.~M\"uller's $\mathsf{iRRAM}$ \cite{Mue01},
P.~Collins' $\mathsf{Ariadne}$ \cite{Col:Ariadne},
B.~Lambov's $\mathsf{RealLib}$ \cite{Lam:RealLib},
and
M. Kone{\v{c}}n{\'{y}}'s $\mathsf{AERN}$ \cite{Kos:AERN}.
More information about real computability can be found e.g.\ in \cite{Bra96,Bra98b,BH98,BHW08,Wei00}.

%
%

\subsection{The Signed-Digit Representation}
\label{sub:SignedDigitRep}

We present an alternative representation for reals inducing the same notion of computability as the Cauchy representation.
It is called the \emph{signed-digit representation}\index{signed-digit representation}.
 
A straightforward version of the signed-digit representation uses the zero-dimen\-sional Polish space $\IZ \times \{\bar{\w{1}},\w{0},\w{1}\}^\IN$ as the space of representatives instead of the Baire space $\Nomega$,
where the symbol $\bar{1}$ stands for the negative digit $-1$.
It is defined by
\[
    \sdrepIZ(z,p):=z + \sum_{i \in \IN} \dfrac{p(i)}{2^{i+1}}
\]
for all integers $z \in \IZ$ and all $p \in \{\bar{\w{1}},\w{0},\w{1}\}^\IN$.
The trick of using the negative digit $\bar{1}$ guarantees that the real numbers are encoded in a reasonable way.
  
We modify $\sdrepIZ$ to obtain a version of the signed-digit representation 
that fits better into our framework which uses the Baire space as fixed representing space.
We employ the effective bijective numbering of the set $\IZ$ of integers given by
\[
   \nu_\IZ(n):=\left\{
    \begin{array}{ll}
      \tfrac{n+1}{2} & \text{if $n$ is odd}
      \\[0.5ex]   
      -\tfrac{n}{2}    & \text{otherwise.}
    \end{array}\right.
\]
So the sequence $(\nu_\IZ(n))_n$ starts with the initial segment $0,1,-1,2,-2,3,-3$.
Now we can define our version of the signed-digit representation $\sdrep\colon \Nomega \parto \IR$ over the Baire space by
\begin{align*}
   \dom(\sdrep) &:=
      \big\{ p \in \IN^\IN \,\big|\,
         p(i) \in \{0,1,2\} \text{ for all $i \geq 1$} \big\},
   \\	 
   \sdrep(p) &:= \nu_\IZ(p(0)) + \sum_{i \geq 1} \dfrac{\nu_\IZ(p(i))}{2^i}\,.
\end{align*}
\noindent
Recall that $2$ encodes the negative digit $-1$.

While inducing the same computability notion as the Cauchy representation (cf.\ Example~\ref{ex:ComparisonRealReps}),
the signed-digit representation has the advantage over the Cauchy representation from Subsection~\ref{sub:CauchyRepReal} of being \emph{proper}. 
This means that the preimage of any compact subset of $\IR$ under $\sdrep$ is compact 
(cf.\ \cite{Sch04,Wei00,Wei03}).
This property makes the signed-digit representation suitable for Type Two Complexity Theory.

%
%

\subsection{Continuous Realizability}
\label{sub:ContinuousRealizable}

In Subsection~\ref{sub:ComputabilityBaireSpace} we have seen that every computable function on the Baire space is continuous with respect to the Baire topology.
Hence the notion of a continuously realizable function plays an important role in investigating the question which representations for a space lead to a reasonable computability notion on that space.

A partial function $f$ between two represented spaces\index{represented space} $\XX=(X,\delta)$ and $\YY=(Y,\gamma)$ 
is called \emph{continuously realizable}\index{continuous realizability} 
(for short \emph{continuous}), 
if there is a partial continuous function $g$ on the Baire space
realizing $f$, meaning that $\gamma g(p)=f\delta(p)$ holds for all $p \in \dom(f\delta)$:
\begin{equation*} 
  \xymatrix{
    X
    \ar[rr]^{\textstyle f}
    \ar@{}[drr]|{\textstyle\pmb{\circlearrowleft}}
    & &
    Y
   \\
    \Nomega
    \ar@{->>}[u]^{\textstyle\delta}
    \ar[rr]_{\textstyle g}
    & &
     \Nomega
     \ar@{->>}[u]_{\textstyle\gamma} 
    } 
\end{equation*}
In this situation $f$ is also called \emph{$(\delta,\gamma)$-continuous} or \emph{continuously realizable} with respect to $\delta$ and $\gamma$.
Continuous realizability of multivariate functions is defined similarly.
Since every computable function on $\Nomega$ is topologically continuous, we obtain a new instance of the slogan: 
\emph{Computability implies Continuity}\index{computability implies continuity}.

\pagebreak[3]
\begin{proposition}
 Any computable function between represented spaces is continuous.
\end{proposition}

\noindent
Hence continuous realizability can be viewed as a topological generalisation of computable realizability.
Computability with respect to an oracle is equivalent to continuous realizability.

%
\subsection{Reducibility and Equivalence of Representations}
\label{sub:reducibility}

Given two representations $\delta_1$ and $\delta_2$ of the same set $X$ we say that 
$\delta_1$ is \emph{computably reducible to} (or \emph{computably translatable into}) $\delta_2$,
if there is a computable function $g$ on the Baire space translating $\delta_1$ into $\delta_2$ in the sense that the diagram
\begin{equation*} 
  \xymatrix{
    & &
    X
   \\
    \Nomega
    \ar@{->>}[urr]^{\textstyle\delta_1}_{\;\;\;\textstyle\pmb{\circlearrowleft}}
    \ar[rr]_{\textstyle g}
    & &
     \Nomega
     \ar@{->>}[u]_{\textstyle\delta_2} 
    } 
\end{equation*}
commutes, meaning $\delta_1(p)=\delta_2(g(p))$ for all $p \in \dom(\delta_1)$.
In this situation we write $\delta_1 \leqcp \delta_2$. 
If $\delta_1 \leqcp \delta_2 \leqcp \delta_1$, then $\delta_1$ and $\delta_2$ are called \emph{computably equivalent} which is written as $\delta_1 \equiv_{\mathrm{cp}} \delta_2$.
Some authors omit the subscript ``$\mathrm{cp}$''.
Topological reducibility (in symbols $\delta_1 \leqt \delta_2$) and
topological equivalence (in symbols $\delta_1 \equiv_{\mathrm{t}} \delta_2$)
is defined analogously by means of \emph{continuous} rather than \emph{computable} translator functions.
Topological reducibility (equivalence) is also known as \emph{continuous reducibility (equivalence)}.
Obviously, $\leqcp,\leqt,\equiv_{\mathrm{cp}},\equiv_{\mathrm{t}}$ are reflexive and transitive.
Computably equivalent representations induce the same notion of computability.

\pagebreak[3]
\begin{proposition}\label{p:translatable}
 Let $\delta_1,\delta_2\colon \Nomega \parto X$ and $\gamma_1,\gamma_2\colon \Nomega \parto Y$ be representations. 
 Let $f\colon X \parto Y$ be a partial function.
\begin{enumerate}
  \item 
   If $\delta_1 \equivcp \delta_2$ and $\gamma_1 \equivcp \gamma_2$, 
   then $f$ is $(\delta_1,\gamma_1)$-computable if, and only if, $f$ is $(\delta_2,\gamma_2)$-computable.
  \item 
   If $\delta_1 \equivt \delta_2$ and $\gamma_1 \equivt \gamma_2$, 
   then $f$ is $(\delta_1,\gamma_1)$-continuous if, and only if, $f$ is $(\delta_2,\gamma_2)$-continuous.
  \item
   If $\delta_1 \equivcp \delta_2$, then $x \in X$ is $\delta_1$-computable if, and only if, $x$ is $\delta_2$-computable.   
  \item
   If $\delta_1 \leqcp \delta_2$, $\gamma_1 \leqcp \gamma_2$ and $f$ is  $(\delta_2,\gamma_1)$-computable, 
   then $f$ is $(\delta_1,\gamma_2)$-computable. 
  \item
   If $\delta_1 \leqt \delta_2$, $\gamma_1 \leqt \gamma_2$ and $f$ is $(\delta_2,\gamma_1)$-continuous, 
   then $f$ is $(\delta_1,\gamma_2)$-continuous. 
  \item
   If $\delta_1 \leqcp \delta_2$ and $x \in X$ is $\delta_1$-computable, then $x$ is $\delta_2$-computable. 
\end{enumerate}
\end{proposition}

\pagebreak[3]
\begin{example}[Comparison of representations for the reals] \label{ex:ComparisonRealReps} \rm \quad\\ \noindent
The decimal representation $\decrep$ is computably translatable into the Cauchy representation $\varrho_\IR$, 
but $\varrho_\IR$ is not even continuously translatable back to $\decrep$.
Nevertheless, computability of a real number with respect to $\decrep$ is equivalent to 
its computability with respect to the Cauchy representation $\varrho_\IR$.
The Cauchy representation is computably equivalent to the signed-digit representation $\sdrep$.
So $\sdrep$ induces the same notion of computability on the Euclidean space as $\varrho_\IR$ by Proposition~\ref{p:translatable}. 
\end{example}

%
\subsection{The Category of Represented Spaces}
\label{sub:RepNomega}
The category of represented spaces\index{represented space} comes in two versions, a continuous one and an effective one.
Both have as objects all represented spaces $(X,\delta)$.
The continuous version, $\RepN$, has as morphisms all total continuously realizable functions between represented spaces.
The effective category $\RepEff$ only uses all total computable functions as morphisms.
The category $\RepN$ is essentially equal to the category $\Mod(\Nomega)$ of modest sets over the Baire space (cf. \cite{Bau01}),
whereas $\RepEff$ corresponds to $\Mod(\Nomega,\Nomega_\sharp)$,
see A.~Bauer's thesis \cite{Bau00} for further information.

%
\subsection{Closure Properties of Represented Spaces}
\label{sub:RepNomega:ClosureProps}

The category $\RepN$ as well as its subcategory $\RepEff$ are known to be 
cartesian closed\index{cartesian closed}
(even locally cartesian closed, cf.\ A.~Bauer's thesis \cite{Bau00}).
We describe how to canonically construct\index{construction of representations} a binary product $\XX \times \YY$ and a function space $\YY^\XX$ in $\RepN$ and identically in $\RepEff$.

Let $\XX=(X,\delta)$ and $\YY=(Y,\gamma)$ be \index{represented space}represented spaces. 
The carrier set of the binary product $\XX \times \YY$ is just the cartesian product $X \times Y$.
The construction of the product representation $\delta \boxtimes \gamma$ employs
the computable bijection $\langle \cdot,\cdot \rangle\colon \Nomega \times \Nomega \to \Nomega$ 
and its computable projections $\pi_1,\pi_2\colon \Nomega \to \Nomega$ from Subsection~\ref{sub:Bairespace}.
One defines \linebreak[3]${\delta \boxtimes \gamma}\colon \Nomega \parto X \times Y$ by
\[
 \delta\boxtimes\gamma\langle p,q \rangle := \big( \delta(p),\gamma(q) \big) \,,
\]
which is a useful short form for
\begin{align*}
 \dom(\delta\boxtimes\gamma) &:=
 \{ r \in \Nomega \,|\, \pi_1(r) \in \dom(\delta),\; \pi_2(r) \in \dom(\gamma) \} \,,
 \\
 \delta\boxtimes\gamma(r) &:= \big( \delta(\pi_1(r)), \gamma(\pi_2(r)) \big)
 \quad\text{for $r \in \dom(\delta\boxtimes\gamma)$.}
\end{align*}
Another familiar notation for $\delta \boxtimes \gamma$ is $[\delta,\gamma]$.

The carrier set of the exponential $\YY^\XX$ in $\RepN$ and in $\RepEff$ consists
of the set $\setC(\XX,\YY)$ of all total continuously realizable functions $f\colon \XX \to \YY$.
Its canonical function space representation $[\delta\to\gamma]$ 
is constructed with the help of the effective representation $\etaFww$ of all partial continuous endofunctions on $\Nomega$ with $G_\delta$-domain from Proposition~\ref{p:etaFww}.
For all $p \in \Nomega$ and total continuous functions $f\colon \XX \to \YY$ one defines
\[
   [\delta\to\gamma](p)=f
    \;:\Longleftrightarrow\;
    \text{the function $\etaFww(p)$ realizes $f$.}
\]
The utm-Theorem for $\etaFww$ ensures that the evaluation function
\[
 \mathit{eval} \colon \YY^\XX \times \XX \to \YY
 \,,\quad
 \mathit{eval}(f,x):=f(x)
\] 
is computable.
The computable smn-Theorem guarantees currying of computable functions:
For any computable $f\colon \ZZ \times \XX \to \YY$ 
the function $\mathit{curry}(f)\colon \ZZ \to \YY^\XX$ defined by
\[ 
  \big(\mathit{curry}(f)(z)\big)(x):=f(z,x)  \;\;\text{for all $z \in \ZZ$ and $x \in X$}  
\]
is computable from $\ZZ$ to $\YY^\XX$;
in particular $\mathit{curry}(f)(z)$ is indeed an element of $\setC(\XX,\YY)$.
Note that $\mathit{curry}(f)(z)$ need not be a computable function from $\XX$ to $\YY$, unless $z$ is a computable element of $\ZZ$.
Likewise, the continuous smn-Theorem ensures that $\mathit{curry}$ maps continuously realizable  functions to continuously realizable functions.
Hence $\YY^\XX$ defined by $(\setC(\XX,\YY),[\delta\to\gamma])$ is an exponential to 
$\XX$ and $\YY$ both in $\RepN$ and in $\RepEff$.
A total function $h$ from $\XX$ to $\YY$ is computable if, and only if,
$h$ is a computable element of~$\YY^\XX$.
We emphasise that the exponential of $\XX$ and $\YY$ in $\RepEff$ coincides with the one in $\RepN$, so that its underlying set may contain incomputable (albeit continuously realizable) functions.

In contrast to its subcategory $\RepEff$, the category $\RepN$ has also countable products and countable coproducts. 
Let $((X_i,\delta_i))_{i \in \IN}$ be a sequence of represented spaces.
One uses the computable projections $\pi_{\infty,i}$ of the countable tupling function for $\Nomega$ from Subsection~\ref{sub:Bairespace} and defines representations
$\mybigboxtimes_{i \in \IN} \delta_i$ for the cartesian product $\prod_{i \in \IN} X_i$ 
and $\mybigboxplus_{i \in \IN} \delta_i$ for the disjoint sum $\bigcup_{i \in \IN} (\{i\} \times X_i)$ by
\begin{align*}
  (\mybigboxtimes\limits_{i \in \IN} \delta_i)(p) &:=
   \big( \delta_0(\pi_{\infty,0}(p)), \delta_1(\pi_{\infty,1}(p)), \delta_2(\pi_{\infty,2}(p)),\dotsc \big)\,,
 \\
  (\mybigboxplus\limits_{i \in \IN} \delta_i) (p)&:= \big( p(0),\delta_{p(0)}(p^{\geq 1})  \big), 
\end{align*}
where $p^{\geq 1}(i):=p(i+1)$.
The ensuing represented spaces
form a countable product and a countable coproduct for $((X_i,\delta_i))_{i \in \IN}$ in $\RepN$, but not necessarily in $\RepEff$.
The represented space $(\{1\} \times X \cup \{2\} \times Y, \delta \boxplus \gamma)$ is a binary coproduct for $\XX$ and $\YY$ in both categories, where the binary coproduct representation $\delta \boxplus \gamma$ is defined analogously to the countable coproduct representation.

We mention the useful notion of the \emph{meet} of represented spaces $((X_i,\delta_i))_i$.
It has the intersection $\bigcap_{i \in \IN} X_i$ as carrier set and the meet representation
$\bigwedge_{i\in\IN} \delta_i$ is defined by
\[
 (\mybigwedge_{i\in\IN} \delta_i)(p)=x \;:\Longleftrightarrow\;
 \forall i \in \IN.\, \delta_i(\pi^\infty_i(p))=x\,.
\]
The representation $\delta \wedge \gamma$ of the binary meet $\XX \wedge \YY$ is defined analogously.

Any subset $M$ of a represented space $\XX=(X,\delta)$ has the corestriction $\delta|^M$ as natural representation which has $\delta^{-1}[M]$ as domain and maps $p \in \delta^{-1}[M]$ to $\delta(p) \in M$.
This entails that both $\RepN$ and $\RepEff$ have equalisers and thus all finite limits.
Moreover, $\RepN$ has all countable limits by being countably cartesian.
Details can be found in A.~Bauer's thesis \cite{Bau00}.


\subsection{Multivalued Functions}
\label{sub:multifunctions}

A \emph{multivalued function}\index{multivalued function}\index{multifunction}\index{problem} 
(or \emph{multifunction} or \emph{operation} or \emph{problem}\index{problem}) is a triple $F=(X,Y,\Graph(F))$, where $\Graph(F)$ is a relation between sets $X$ and $Y$, written as $F\colon X \rightrightarrows Y$.
The set $X$ is called the \emph{source} and $Y$ the \emph{target} of $F$.
The intuition behind this definition is that
the elements of $F[x]:=\{ y \in Y \,|\, (x,y) \in F\}$ are viewed as the legitimate results for an input $x$ under the multivalued function $F$. 
If the domain $\dom(F):=\{x \in X\,|\, F[x] \neq \emptyset\}$ of $F$ is equal to $X$, then $F$ is called \emph{total}.
Ordinary total or partial functions $f\colon X \parto Y$ can be seen as multivalued functions such that the image $f[x]$ is a singleton for every $x$ in the domain $\dom(f)$ of $f$.

A multivalued function $F$ between represented space $\XX=(X,\delta)$ and $\YY=(Y,\gamma)$ is called \emph{computable}\index{computable multifunction}\index{computable problem}\index{computable realizability}
(or \emph{computably realizable} or \emph{$(\delta,\gamma)$-computable}),
if there is a computable function $g$ on $\Nomega$ such that
\[
 \gamma(g(p)) \in F[\delta(p)] \quad \text{for every $p\in \dom(\delta)$ with $\delta(p) \in \dom(F)$.}
\]
So a realizer $g$ for $F$ picks a possible result $y \in F[x]$ for an input $x$ which may heavily depend on the input name $p \in \delta^{-1}[x]$. 
From the point of view of logic, a computable realizer for $F$ can be interpreted as a computable realizer of the $\forall\exists$-statement
\[
  \forall x \in \dom(F). \exists y \in Y.\, (x,y) \in \Graph(F) \,.
\]
Similarly, $F$ is called \emph{continuously realizable}\index{continuous realizability} (for short \emph{continuous}), if it has a continuous realizer $g$ satisfying the above formula.

\pagebreak[3]
\begin{example}[Finite precision tests] \rm \quad\\ \noindent 
A powerful computable multivalued function is the finite precision test\index{finite precision test}
$T\colon \IR^3 \rightrightarrows  \IN$ defined by
\[
 T[x,y,\varepsilon]:= (x<_\varepsilon y):=\left\{
  \begin{array}{cl}
    \{1\}   & \text{if $x \leq y-\epsilon$}
   \\
    \{0,1\} & \text{if $y-\varepsilon < x < y$}
   \\ 
    \{0\}   & \text{if $x \geq y$}
  \end{array} \right.
\]
for all $x,y,\varepsilon \in \IR$.
This $(\varrho_\IR\boxtimes\varrho_\IR\boxtimes\varrho_\IR, \varrho_\IN)$-computable test 
(where $\varrho_\IN(p):= p(0)$)
is a basic operation in P.~Hertling's and V.~Brattka's feasible real random access machine (cf. \cite{BH98}).
By contrast, the ordinary precise test ``$\mathrm{x}<\mathrm{y}$\,?'' is incomputable by being discontinuous.

In his package \textsf{iRRAM}\index{iRRAM} for exact real number computation (cf.\ \cite{Mue01}),
N.~M\"uller has implemented a similar multivalued test $\mathit{positive}\colon \IR \times \IZ  \rightrightarrows \IN$.
It is defined by
\[
 \mathit{positive}[x,z]:=\left\{
  \begin{array}{cl}
    \{1\}   & \text{if $x>  2^{z}$}
   \\
    \{0,1\} & \text{if $-2^{z} \leq x \leq 2^{z}$}
   \\ 
    \{0\}   & \text{if $x < -2^{z}$}
  \end{array} \right.
\]
for all $x \in \IR$, $z \in \IZ$. 
The multivalued function $\mathit{positive}$ is $(\varrho_\IR \boxtimes \varrho_\IZ, \varrho_\IN)$-computable.
The unsharp positivity test can often be employed as a computable substitute for the discontinuous, hence incomputable test ``$\mathrm{x}>0$\,?''.
For a further discussion of unsharp tests we refer to \cite{BH98,Mue01}.
\end{example}

From the computability point of view, the most appropriate definition of composition of two multivalued functions is given by
\begin{align*}
 \dom(G\circ F)&:=\big\{ x \in \dom(F) \,\big|\, F[x] \subseteq \dom(G) \big\}\,,
 \\
 G \circ F[x]  &:= \big \{ z \,\big|\, \exists y \in F[x].\, z \in G[y] \big\} \text{ for $x \in \dom(G \circ F)$.}
\end{align*}
This definition guarantees that the class of computable multivalued functions is closed under composition (see \cite{Bra03}).
The seemingly more natural version of composition that has $\{ x \in \dom(F) \,|\, F[x] \cap \dom(G) \neq \emptyset \}$ as domain does not preserve computability of multivalued functions.

\subsection{Multirepresentations}
\label{sub:multireps}

A \emph{multirepresentation}\index{multirepresentation} $\delta$ of a set $X$ is a partial multivalued function from the Baire space to $X$ 
which is surjective in the sense that for every $x \in X$ there exists some $p \in \Nomega$ such that $x \in \delta[p]$. 
So multirepresentations generalise representations by allowing any $p \in \Nomega$ to be a name of more than one element. 
Such a pair $(X,\delta)$ is called a \emph{multirepresented space}\index{multirepresented space}.

A partial function $f\colon X \parto Y$ between multirepresented spaces $\XX=(X,\delta)$ and $\YY=(Y,\gamma)$
is called \emph{computable}\index{computable function} or \emph{computably realizable}\index{computable realizability} or \emph{$(\delta,\gamma)$-computable},
if there is a computable function $g\colon \Nomega \parto \Nomega$ satisfying
\[
  x \in \delta[p] \;\&\; x \in \dom(f) \implies f(x) \in \gamma[g(p)] 
  \,;
\]
similarly for multivariate functions.
Continuous realizability\index{continuous realizability} w.r.t.\ multirepresentations is defined via \emph{continuous} realizers
rather than computable realizers in the spirit of Subsection~\ref{sub:ContinuousRealizable}.
The category $\MRep$ of multirepresented spaces $(X,\delta)$ and continuously realizable total functions is equivalent to the category $\mathsf{Asm}(\Nomega)$ of assemblies \cite{LN15} over the Baire space.
Thus $\MRep$ is locally cartesian closed, countably complete and countably cocomplete.

There exist constructors for forming canonical multirepresentations for products, coproducts, meets, and function spaces generalising those for representations presented in Subsection~\ref{sub:RepresentedSpaces}.
Details can be found in \cite{Sch:phd}.

\pagebreak[3]
\begin{example}[The space of partial continuous functions]\label{ex:spacePP}
 \rm \quad\\ \noindent
 Let $\XX=(X,\delta)$, $\YY=(Y,\gamma)$ and $\ZZ$ be (multi-)represented spaces.
 A canonical multirepresentation for the set of partial continuously realizable functions from $\XX$ to $\YY$ is defined by
 \[
  f \in [\delta{\parto}\gamma][p] \;:\Longleftrightarrow\;
  \forall (r,x) \in \Graph(\delta). \big(x \in \dom(f) \implies f(x) \in \gamma[\etaFww(p)(r)] \big)\,,  
 \]
 where $\etaFww$ is taken from Proposition~\ref{p:etaFww}.
 We denote by $\PP(\XX,\YY)$ the ensuing multirepresented space.
 The evaluation function $\mathit{eval}\colon \PP(\XX,\YY) \times \XX \parto \YY$ is computable.
 Moreover, $[\delta\parto\gamma]$ admits currying of the following kind:
 for all partial computable functions $h\colon \ZZ \times \XX \parto \YY$  the total function $\mathit{curry}(h)\colon \ZZ \to \PP(\XX,\YY)$ defined by
 \[
  \dom\big(\mathit{curry}(h)(z)\big) :=\{ x \,|\, (z,x) \in \dom(h)\}  ,\;
  \big(\mathit{curry}(h)(z)\big)(x) :=h(z,x) 
 \]
 is computable as well. 
\end{example}

More information can be found in \cite{Sch:phd,WG09}.


\section{Admissible Representations for Topological Spaces}
\label{sec:admissibility}

We introduce and study admissibility as a notion of topological well-behavedness for representations\index{representation} for topological spaces\index{topological space}.
This notion is a generalisation of a similar notion introduced by C.~Kreitz and K.~Weihrauch in \cite{KW85}
which applies to representations for second-countable topological spaces.

%
\subsection{The Topology of a Represented Space} 
\label{sub:QuotientTopology}
Any represented space $\XX=(X,\delta)$ carries a natural topology, 
namely the \emph{final topology} or \emph{quotient topology} induced by the partial surjection $\delta$.
The final topology is defined as the family of sets
\[
 \tau_\delta:=
 \big\{ U \subseteq X \,\big|\, \delta^{-1}[U]=W \cap \dom(\delta) \text{ for some open subset $W \subseteq \Nomega$}\big\}\,. 
\]
It is equal to the finest (= largest) topology on $X$ such that $\delta$ is topologically continuous.
We denote by $\calT(X,\delta)$ the topological space that carries the final topology $\tau_\delta$
and call it the \emph{associated} space. 
The following fact is crucial.

\pagebreak[3]
\begin{proposition}\label{p:contrel2topcontinuous}
 Let $\XX$ and $\YY$ be represented spaces.
 Any total continuously realizable function $f\colon \XX \to \YY$ 
 is topologically continuous with respect to the final topologies 
 of the respective representations.
\end{proposition}

\noindent
Hence the operator $\calT$ 
constitutes a forgetful functor $\calT\colon \RepN \to \Top$ mapping continuously realizable functions to themselves.
In Subsections~\ref{sub:SequentialSpaces} and~\ref{sub:QCB} we will see that $\calT(X,\delta)$ is a sequential topological space and even a qcb-space. 

\pagebreak[3]
\begin{proposition}[Computability implies continuity]
\label{p:comprel2topcontinuous}
 Any total computable function between two represented spaces is 
 topologically continuous\index{computability implies continuity}
 with respect to the topologies asscociated to the represented spaces.
\end{proposition}

%
\subsection{Sequential Topological Spaces}
\label{sub:SequentialSpaces}

A \emph{sequential space}\index{sequential space} (cf.\ \cite{ELS04,Eng89,Hyl79}) is a topological space $\WX$
such that every sequentially closed subset of $\WX$ is closed in $\WX$.
A subset $A$ of $\WX$ is called \emph{sequentially closed}, 
if $A$ contains every limit $x$ of every sequence $(x_n)_n$ in $A$.
Remember that $(x_n)_n$ is defined to converge in a topological space $\WX$ to an element $x$
(in symbols $(x_n)_n \To{\WX} x$ or shorter $(x_n)_n \to x$),
if any open set $U \ni x$ contains $(x_n)_n$ eventually,
meaning that there is some $n_0 \in \IN$ with $x_n \in U$ for all $n \geq n_0$.
The complements of sequentially closed sets are called \emph{sequentially open}. 

Basic examples of sequential spaces are metrisable spaces and second-countable spaces.
The following fact is crucial:

\pagebreak[3]
\begin{example}[Sequential spaces]\label{ex:finaltopology:sequential} \rm \quad\\
\noindent
 For any represented space $\XX$, 
 the associated topological space $\calT(\XX)$ is sequential.
\end{example}

Neither forming the topological product $\WX \times_\Top \WY$
nor taking the compact-open topology on the family $\setC(\WX,\WY)$ of total topologically continuous functions from $\WX$ to $\WY$ nor taking a topological subspace preserve the property of being sequential.
For example, the compact-open topologies on $\setC(\Nomega,\IN)$ and on $\setC(\IR^\IR,\IR)$ are not sequential.

A partial function $f\colon \WX \parto \WY$ between topological spaces is called \emph{sequentially continuous}\index{sequentially continuous},
if $f$ preserves convergent sequences,
i.e., $(x_n)_n \To{\WX} x$ implies $(f(x_n))_n \To{\WY} f(x)$, whenever the sequence $(x_n)_n$
and its limit $x$ are in the domain of~$f$.
The map $f$ is called \emph{topologically continuous}\index{topologically continuous}, 
if for every open set $V$ in $\WY$ there is some open set $U$ in $\WX$ such that $f^{-1}[V]=U \cap \dom(f)$.
Topological continuity implies sequential continuity, but not vice versa.
Importantly, sequential continuity coincides with topological continuity
for total unary functions between sequential spaces. 
The reader should be warned that this equivalence does not hold in general for multivariate functions nor for non-total functions.

By $\Seq$ we denote the category of sequential topological spaces as objects
and total sequentially continuous functions as morphisms.
This category is known to be cartesian closed (see \cite{ELS04,Hyl79}).
By contrast, the category $\Top$ of topological spaces and total topologically continuous functions fails to be cartesian closed.


\subsection{Sequentialisation of Topological Spaces}
\label{sub:Sequentialisation}

The family of all sequentially open sets of a topological space $\WX$ forms a (possibly) finer topology on the carrier set of $\WX$.
It is called the \index{sequentialisation}\emph{sequentialisation} of the topology of $\WX$.
The topological space $\seq(\WX)$ carrying this topology
is the unique sequential space that has the same convergence relation on sequences as $\WX$.
It is equal to $\WX$ if, and only if, $\WX$ is sequential.
The functor $\seq$ is right-adjoint to the inclusion functor of $\Seq$ into the category $\Top$ of topological spaces (see \cite{Hyl79}).
So $\seq(\WX)$ is referred to as the \emph{sequential coreflection} of~$\WX$.

%
\subsection{Topological Admissibility} 
\label{sub:admissibility}
Given represented spaces, the natural question arises which functions are computable between them.
Conversely, one wants to know how to construct appropriate representations that admit computability of a relevant family of functions and thus have certain required effectivity properties.

Since computability w.r.t.\ given representations implies continuity w.r.t.\ these representations,
it is reasonable to consider two different kinds of effectivity properties, the topological ones and the purely computable ones.
Many representations fail to be appropriate simply by having unsuitable topological properties preventing some interesting functions from being at least continuously realizable. 
The prime example is the notorious decimal representation $\decrep$:
an inspection of the proof of Example~\ref{ex:MultiByThree} shows that multiplication by $3$ is in fact not even continuously realizable.
So the misbehaviour of the decimal representation is of purely topological nature.

Therefore we turn our attention to continuous realizability
and identify those representations for topological spaces which induce an appropriate class of continuously realizable functions.
This is done by introducing the property of \emph{admissibility} for representations. 
It aims at guaranteeing that every topologically continuous function is continuously realizable.
Definition~\ref{def:admissible:Top} generalises a previous notion of admissibility introduced by C.~Kreitz and K.~Weihrauch 
(see \cite{KW85,Wei00} and Subsection~\ref{sub:KWadmissible}) which is taylor-made for topological spaces with a countable base.

\pagebreak[3]
\begin{definition}[Admissible representations \cite{Sch02}] \rm
\label{def:admissible:Top} \quad\\ \noindent
 A representation $\delta\colon \Nomega \parto \WX$ is called \emph{admissible}\index{admissible representation} for a topological space $\WX$, if
 \begin{itemize} 
  \item[(a)] 
   $\delta$ is continuous,
  \item[(b)]
   for every partial continuous function $\phi\colon \Nomega \parto \WX$
   there is a continuous function $g\colon \Nomega \parto \Nomega$ 
   satisfying $\phi=\delta \circ g$, i.e., the diagram
   $$
    \xymatrix{
      & &
      \WX
     \\
      \dom(\phi)
      \ar[rr]_{\textstyle g}   
      \ar@{->>}[rru]^{\textstyle\phi}_{\;\;\;\textstyle\pmb{\circlearrowleft}}
      & &
       \dom(\delta)
       \ar@{->>}[u]_{\textstyle\delta} 
    } 
  $$ 
  commutes.
 \end{itemize}
\end{definition}

\noindent
The condition (b) is referred to as the \emph{maximality property} of $\delta$.
It postulates that any continuous representation of $\WX$ can be translated into $\delta$. 
Any admissible representation for $\WX$ is maximal in the class of  
continuous representations of $\WX$ with respect to continuous reducibility $\leqt$.
Continuous reducibility (= topological reducibility) is defined by:
$\phi \leqt \phi'$ iff there is a continuous $g\colon \dom(\phi) \to \Nomega$ with $\phi=\phi' \circ g$. 

\pagebreak[3]
\begin{proposition}\label{p:adm:leqtmaximal}
 Let $\delta$ be a representation of a topological space $\WX$.
 Then $\delta$ is admissible if, and only if, it is $\leqt$-maximal 
 in the class of continuous representations for~$\WX$.
\end{proposition}

Admissibility ensures the equivalence of mathematical continuity and continuous realizability as desired.

\pagebreak[3]
\begin{theorem}\label{th:adm:relcont=cont}
 Let $\delta$ and $\gamma$ be admissible representations for topological spaces $\WX$ and~$\WY$, respectively.
 Then a partial function $f\colon \WX \parto \WY$ is sequentially continuous
 if, and only if,
 $f\colon (\WX,\delta) \parto (\WY,\gamma)$ is continuously realizable.
\end{theorem}

\noindent
Recall that topological continuity implies sequential continuity, but not vice versa.

Admissibility is necessary only for the representation of the target space, 
whereas the source representation is merely required to be a quotient representation\index{quotient representation}
($\delta$ is called a \emph{quotient representation} for a topological space $\WX$,
 if its final topology is equal to the topology of $\WX$).
We formulate this fact with the help of the notion of an \emph{admissibly represented space}\index{admissibly represented space}:
this is a represented space $\YY$ such that its representation is admissible for the associated topological space $\calT(\YY)$.  
A.~Pauly coined this notion in \cite{Pau16} using a different definition
which is equivalent to the above definition by Theorem~\ref{th:deltaIS:admissible}.

\pagebreak[3]
\begin{theorem}
\label{th:maintheo}
 Let $\XX$ be a represented space and let $\YY$ be an admissibly represented space.
 Then a total function $f\colon \XX \to \YY$ is continuously realizable
 if, and only if,
 $f$ is topologically continuous between the associated topological space $\calT(\XX)$ and $\calT(\YY)$.
\end{theorem}

For partial functions and for multivariate functions we only get equivalence with sequential continuity.

\pagebreak[3]
\begin{theorem}
\label{th:maintheo:multivariate}
 Let $\XX_i$ be represented spaces for $i\in \{1,\dotsc, k\}$ and let $\YY$ be an admissibly represented space.
 Then a partial function $f\colon \XX_1 \times\dotsc\times \XX_k \parto \YY$ is continuously realizable\index{continuous realizability},
 if, and only if,
 $f$ is sequentially continuous\index{sequentially continuous} 
 between the associated topological spaces $\calT(\XX_1),\dotsc,\calT(\XX_k)$
 and $\calT(\YY)$.
\end{theorem}
\noindent
Since the topological product of the spaces $\calT(\XX_i)$ is not necessarily sequential, continuous realizability only implies sequential continuity.

Admissible representations have the following quotient properties.

\pagebreak[3]
\begin{proposition}
\label{p:adm:quotient}
Let $\delta$ be an admissible representation for a topological space $\WX$.
\begin{enumerate}
 \item 
  $\delta$ lifts convergent sequences, 
  i.e., for every sequence $(x_n)_n$ converging in $\WX$ to some $x \in \WX$ 
  there is a sequence $(p_n)_n$ converging in $\Nomega$ to some $p \in \Nomega$
  such that $\delta(p_m)=x_m$ for all $m \in \IN$ and $\delta(p)=x$.
 \item \label{en:admquotient}
  $\delta$ is a quotient representation for $\WX$
  if, and only if, $\WX$ is a sequential space.
 \item
  $\delta$ is an admissible quotient representation for the sequential coreflection of $\WX$.
\end{enumerate}
\end{proposition}

Though the decimal representation $\decrep$ is a quotient representation for the Euclidean space~$\IR$, 
it does not lift convergent sequences. 
Hence it is not admissible for $\IR$.
By contrast, the Cauchy representation enjoys admissibility.
This explains why the latter induces a useful notion of computability, whereas the former fails to do so.

\pagebreak[3]
\begin{example}\label{ex:sdrep:adm} \rm
 The Cauchy representation and the signed-digit representation are admissible for the Euclidean space $\IR$, 
 whereas the decimal representation is not.   
\end{example}

Proofs of this subsection can be found in \cite{Sch:phd}.

%
\subsection{Examples of Admissible Representations}
\label{sub:ExamplesAdmissible}

We present admissible representations for some relevant spaces.

\pagebreak[3]
\begin{example}[Admissible representations for relevant metric spaces] \rm \quad \label{ex:adm:reps} 
\begin{enumerate} \vspace*{-1ex}
 \item
  The identity on $\Nomega$ is an admissible representation for $\Nomega$. 
 \item
  For the discrete natural numbers $\IN$
  we use as canonical admissible representation $\varrho_\IN$
  the one defined by $\varrho_\IN(p):=p(0)$.
 \item
  The one-point compactification $\INu$ has as carrier set $\IN \cup \{\infty\}$.
  Its topology has as basis the sets $\{a\}$ and $\{ \infty, n \,|\, n \geq a \}$ for all $a \in \IN$.
  An admissible representation $\varrho_{\INu}$ for $\INu$ is defined by
  \[
   \varrho_{\INu}(p):=\left\{
    \begin{array}{cl}
      \infty & \text{if $p=\w{0}^\omega$}
      \\[0.5ex]   
      \min\{n \in \IN \,|\, p(n) \neq 0\}  & \text{otherwise.}
    \end{array}\right.
  \]
  for all $p \in \Nomega$.
 \item
  Any \index{Polish space}Polish space $\WX$ has a total admissible representation.
  We briefly sketch a construction.
  Let $d\colon \WX \times \WX \to \IR$ be a complete metric
  inducing the topology of $\WX$,
  and let $\{a_i \,|\, i \in \IN\}$ be a dense subset of $\WX$.
  Using the continuous function $\ell\colon \Nomega \to \INu$ given by
  \[
    \ell(p):=
    \min\big( \{\infty\} \cup \{m \in \IN \,|\, d(a_{p(m)},a_{p(m+1)})\geq 1/2^{m+1} \} \big),
  \]
  we define a total representation $\delta\colon\IN^\IN \to \WX$ by
  \[
   \delta(p):= \left\{
    \begin{array}{cl}
     \lim\limits_{i \to \infty} a_{p(i)}
     & \text{if $d(a_{p(m)},a_{p(m+1)}) < 1/2^{m+1}$ for all $m \in \IN$}
     \\
     a_{p(\ell(p))} & \text{otherwise.}
    \end{array}\right.  
  \]
  Since $d( \delta(p), a_{p(m)} ) \leq 2^{-m}$
  holds for all $p \in \IN^\IN$ and $m \leq \ell(p)$, $\delta$~is continuous.
  The maximality property follows from the fact that the restriction of $\delta$
  to $\ell^{-1}\{\infty\}$ is a version of the Cauchy representation for $\WX$
  similar to the admissible representation defined in \cite[Definition 3.4.17]{Wei00}. 
\end{enumerate}  
\end{example}

\noindent
We turn our attention to non-Hausdorff examples.
  
\pagebreak[3]  
\begin{example}[Admissible representations for relevant non-$T_2$-spaces] 
\rm \quad 
\label{ex:adm:reps:NonT2}\label{ex:Sierpinski}\label{ex:EnRep}\label{ex:deltabeta}
\begin{enumerate} \vspace*{-1ex}
 \item
  The Sierpi{\'n}ski space $\IS$ has $\{\bot,\top\}$ as its carrier set
  and is topologised by the family $\{ \emptyset,\{\top\},\{\bot,\top\} \}$.
  Note that the constant sequence $(\top)_n$ converges both to $\top$ and to $\bot$.
  A total admissible representation $\varrho_\IS$ for $\IS$ is defined by
  \[
   \varrho_\IS(p) := \left\{
   \begin{array}{ll}   
    \bot &\text{if $p = \w{0}^\omega$}
    \\
    \top &\text{otherwise,}
   \end{array}\right. 
  \]
  where $\w{0}^\omega$ denotes the constant sequence consisting of $\w{0}$'s.
 \item
  Scott's graph model $\Pomega$ is a prime example of a Scott domain.
  Its carrier set is the power set of $\IN$ ordered by set inclusion $\subseteq$.
  A countable basis of the Scott topology on $\Pomega$ is given by the family of sets
  $\{ M \subseteq \IN \,|\, E \subseteq M \}$,
  where $E$ varies over the finite subsets of~$\IN$.
  The enumeration representation $\En\colon \Nomega \to \Pomega$ defined by
  \[
   \En(p):=\{ p(i)-1 \,|\, i \in \IN \;\&\; p(i)>0 \} 
  \]
  is an admissible representation for $\Pomega$ viewed as a topological space endowed with the Scott topology (cf.\ \cite{Wei00}).
 \item
  The \emph{lower reals} $\IR_<$ are endowed with the topology $\{ (a;\infty) \,|\, a \in \IR \} \cup \{ \emptyset, \IR\}$.
  An admissible representation $\delta_<$ for $\IR_<$ is given by: 
  $\delta_<(p)=x \,:\Longleftrightarrow\, \En(p)=\{ i \,|\, \nu_\IQ(i)< x \}$,
  where $\nu_\IQ$ is a canonical numbering of the rational numbers.
 \item
  Any second-countable $T_0$-space $\WX$ has an admissible representation (cf. \cite{KW85}).
  Given a countable subbase $\{ \beta_i \,|\, i \in \IN\}$ with $\bigcup_{i\in\IN} \beta_i=X$, one defines $\delta_\beta\colon \Nomega \parto X$ by
  \[  
    \delta_\beta(p)=x \;:\Longleftrightarrow\;
     \En(p) = \{ i \in \IN \,|\, x \in \beta_i \} \,.
  \]
  So $p$ is a name of $x$ if, and only if, $p$ lists all (indices of) subbasic sets which contain $x$.
  Remember that $T_0$-spaces are exactly those topological spaces $\WY$ 
  $\eta_\WY\colon y \mapsto \{ V \text{ open in $\WY$} \,|\, y \in V \}$
  is injective, ensuring that $\delta_\beta$ is indeed well-defined.
  C.~Kreitz and K.~Weihrauch call $\delta_\beta$ a \emph{standard representation} for $\WX$ in \cite{KW85}
  and showed that $\delta_\beta$ is continuous and has the maximality property. 
  Hence $\delta_\beta$ is admissible in the sense of Definition~\ref{def:admissible:Top}.
 \item
  M.~de Brecht has characterised the second-countable $T_0$-spaces
  enjoying a total admissible representation as the quasi-Polish spaces (see \cite{dBre13}).  
  Quasi-Polishness is a completeness notion for quasi-metric spaces.
  For example, Scott's graph model $\Pomega$ is quasi-Polish.
\end{enumerate}
\end{example}

%
\subsection{The Admissibility Notion by Kreitz and Weihrauch}
\label{sub:KWadmissible}
C.~Kreitz and K.~Weihrauch introduced a notion of admissible representations  taylor-made for representing $T_0$-spaces $\WX$ with a countable base.
They call a representation $\phi$ for $\WX$ \emph{admissible},
if $\phi$ is topologically equivalent to the standard representation $\delta_\beta$ presented in Example~\ref{ex:adm:reps:NonT2}.
This notion of admissibility guarantees equivalence of continuous realizability 
and topological continuity for partial multivariate functions between second-countable $T_0$-spaces.
This equivalence is referred to as \index{Main Theorem of Computable Analysis}\emph{Main Theorem} in \cite{Wei87,Wei00}.  
Moreover, C.~Kreitz and K.~Weihrauch proved in \cite{KW85} that $\phi$ is topologically equivalent to $\delta_\beta$
if, and only if, $\phi$ is continuous and has the maximality property.
Therefore this admissibility notion agrees with the one in Definition~\ref{def:admissible:Top} for representations of second-countable $T_0$-spaces.

%
\subsection{Constructing Admissible Representations}
\label{sub:construction:adm}

In Subsection~\ref{sub:RepNomega:ClosureProps} we have seen how to form products, coproducts, meets and function spaces in $\RepN$.
The corresponding constructions\index{construction of representations} preserve admissibility.

\pagebreak[3]
\begin{proposition}
\label{p:prodrep:adm}
 Let $\delta_i$ be an admissible representation for a topological space $\WX_i$.
\begin{enumerate} 
 \item
  $\delta_1 \boxtimes \delta_2$ is an admissible rep.\ for the topological product of $\WX_1$ and $\WX_2$.
 \item
  $\mybigboxtimes_{i\in \IN} \delta_i$ is an admissible representation for the Tychonoff product $\prod_{i \in \IN} \WX_i$.
 \item
  $\delta_1 \boxplus \delta_2$ is an admissible representation for the binary coproduct $\WX_1 \uplus \WX_2$.
 \item
  $\mybigboxplus_{i\in \IN} \delta_i$ is an admissible representation for the countable coproduct $\biguplus_{i \in \IN} \WX_i$.  
 \item
  $\delta_1 \wedge \delta_2$ is an admissible representation for the meet $(X_1 \cap X_2,\setO(\WX_1) \wedge \setO(\WX_2))$.
 \item
  $\mybigwedge_{i\in \IN} \delta_i$ is an admissible representation for the meet $(\bigcap_{i \in \IN} \WX_i, \bigwedge_{i \in \IN} \setO(\WX_i))$.   
\end{enumerate}
\end{proposition} 

\noindent
Here $\bigwedge_{i \in \IN} \setO(\WX_i)$ denotes the countable meet topology which is generated by the basis
\[
 \big\{ \bigcap\nolimits_{i=0}^k U_i \,\big|\, k \in \IN,\; \text{$U_i$ open in $\WX_i$ for all $i \leq k$} \big\}
 \,.
\]
The binary meet topology $\setO(\WX_1) \wedge \setO(\WX_2)$ is defined analogously.
 
For the function space representation to be admissible
it suffices that the target space representation is admissible,
whereas the source space representation only has to have appropriate quotient properties.
This fact is crucial in showing that every qcb$_0$-space has an admissible representation (cf. Subsection~\ref{sub:qcb:admissible}).
Recall that any admissible representation lifts convergent sequences;
moreover, it is a quotient representation, whenever the represented topological space is sequential
(cf.\ Proposition~\ref{p:adm:quotient}).
 
\pagebreak[3] 
\begin{proposition} 
\label{p:funcrep:adm}
 Let $\gamma$ be an admissible representation for a topological space $\WY$.  
\begin{enumerate}
 \item
  Let $\delta$ be a quotient representation for a topological space $\WX$.
  Then $[\delta\to\gamma]$ is an admissible representation for the space of all total topologically continuous functions from $\WX$ to $\WY$ endowed with the compact-open topology.
 \item
  Let $\delta$ be a continuous representation of a topological space $\WX$ that lifts convergent sequences.
  Then $[\delta\to\gamma]$ is an admissible representation for the space of all total sequentially continuous functions from $\WX$ to $\WY$ endowed with the compact-open topology.
\end{enumerate}  
\end{proposition} 

Now we turn our attention to subspaces and sequential embeddings.
A function $e\colon \WZ \to \WX$ is called a \emph{sequential embedding} of $\WZ$ into $\WX$,
if $e$ is injective and satisfies 
$(z_n)_n \To{\WZ} z \Longleftrightarrow (e(z_n))_n \To{\WX} e(z)$.

\pagebreak[3]
\begin{proposition}
\label{p:seqembedding:adm}
 Let $\delta$ be an admissible representation for a topological space $\WX$.
\begin{enumerate}
 \item
  Let $\WY$ be a topological subspace of\/ $\WX$.
  Then the corestriction $\delta|^\WY$ is an admissible representation for $\WY$.
 \item 
  Let $e\colon \WZ \to \WX$ be a sequential embedding of a topological space $\WZ$ into $\WX$.
  Then $\zeta\colon \Nomega \parto \WZ$ defined by
  $ \zeta(p)=z \;:\Longleftrightarrow\; e(z)=\delta(p)$
  is an admissible representation for $\WZ$.
\end{enumerate}
\end{proposition}

For proofs we refer to \cite{Sch:phd,Sch02}.

\subsection{Pseudobases}
\label{sub:charac:Top:adm}

The class of topological spaces enjoying an admissible representation can be characterised with the help of the notion of a pseudobase.
A family $\calB$ of subsets of a topological space $\WX$ is called a \emph{pseudobase}\index{pseudobase} for $\WX$,
if for every open set $U$ in $\WX$ and for every sequence $(x_n)_n$ converging to some element $x \in U$
there is a set $B \in \calB$ with
 \[
   x\in B \subseteq U
   \quad\text{and}\quad x_n \in B \quad \text{for almost all $n \geq n_0$. }
 \]  
A family $\calB$ of subsets such that its closure under finite intersection
is a pseudobase for $\WX$ is called a \emph{pseudosubbase} for $\WX$.
Any base of a topological space is a pseudobase; the converse is not true in general. 
Pseudobases become interesting when they are countable.
Given pseudobases for spaces $\WX$ and $\WY$, there exist canonical constructions for forming pseudobases
for product, coproduct, meet, exponentiation, and subspaces (cf.\ \cite{ELS04,Sch:phd,Sch02}).

\pagebreak[3]
\begin{example}[Pseudobase generated by an admissible representation]
\label{ex:adm2psb}
\rm \quad \noindent
 Let $\delta$ be an admissible representation for a topological space $\WX$.
 Then the sets $\delta[u\Nomega \cap \dom(\delta)]$, where $u$ varies over the finite strings over $\IN$, 
 form a countable pseudobase for $\WX$.
\end{example}

Together with the $T_0$-property, the existence of a computable pseudobase characterises the class of admissibly representable topological spaces.

\pagebreak[3]
\begin{proposition}[Topological spaces with admissible representations] \quad
\label{p:adm:pseudobase}
\begin{enumerate} \vspace*{-1ex}
 \item 
  A topological space has an admissible representation 
  if, and only if,
  it has a countable pseudobase and satisfies the $T_0$-property.
 \item
  A topological space has an admissible quotient representation 
  if, and only if,
  it is sequential, has a countable pseudobase and satisfies the $T_0$-property.
\end{enumerate}
\end{proposition}

Given a countable pseudobase $\{ \beta_i \,|\, i \in \IN\}$ for a $T_0$-space $\WX$, 
an admissible representation $\theta_\beta$ for $\WX$ can be constructed by
\[
  \theta_\beta(p) = x \;:\Longleftrightarrow\; \left\{
  \begin{array}{l}
    \En(p) \subseteq \{ i \in \IN \,|\, x \in \beta_i \}
    \quad\&
    \\[1ex]
     \forall U \in \eta_\WX(x). \exists i \in \En(p). \beta_i \subseteq U
   \end{array} \right.
\]
where $\eta_\WX(x)$ denotes the family of open sets containing $x$ and $\En$ is the enumeration numbering of $\{ M \,|\, M \subseteq \IN\}$ from Example~\ref{ex:EnRep}.
The $T_0$-property guarantees that at most one element $x$ satisfies the right hand side of the displayed formula.
On the other hand, a simple cardinality argument shows that the existence of an admissible representation implies the $T_0$-property.

 
\section{Qcb-Spaces}
\label{sec:QCB}

We discuss the class of qcb-spaces\index{qcb-space} and their subclass of qcb$_0$-spaces.
Qcb$_0$-spaces are shown to be exactly the class of topological spaces that have an admissible quotient representation.
Their excellent closure properties exhibit them as an appropriate topological framework for studying computability on spaces with cardinality of the continuum.
Moreover, we investigate computability properties of qcb-spaces using the notion of an effective qcb-space.
Finally, we present quasi-normal spaces as an important subclass of Hausdorff qcb-spaces suitable for studying computability in Functional Analysis.


\subsection{Qcb-Spaces}
\label{sub:QCB}

A topological space $\WX$ is called a \emph{qcb-space}\index{qcb-space}
(\emph{q}uotient of a \emph{co}untably-\emph{b}ased space),
if $\WX$ can be exhibited as a topological quotient of a topological space $\WW$ with a countable base.
This means that there exists a surjection $q\colon \WW \twoheadrightarrow \WX$ such that
$\WX$ carries the finest (= largest) topology $\tau_q$ such that $q$ is topologically continuous
(in this case $q$ is called a \emph{quotient map}).
Clearly, this finest topology is equal to $\{ U \subseteq \WX \,|\, q^{-1}[U] \text{ open in } \WW \}$.
By being a topological quotient of a sequential space, any qcb-space is sequential as well (cf. \cite{Eng89}).
Recall that this implies that a total unary function between qcb-spaces is topologically continuous if, and only if, it is sequentially continuous.
We write $\QCB$ for the category of qcb-spaces with continuous functions as morphisms
and $\QCBZ$ for the full subcategory of \emph{qcb$_0$-spaces},
which are qcb-spaces satisfying the $T_0$-property.

Clearly, any separable metrisable space is a qcb$_0$-space by being second-count\-able and Hausdorff.
Moreover, the topological space associated to a represented space is a qcb-space. 
So henceforth we consider the functor $\calT$ from Subsection~\ref{sub:QuotientTopology} as a functor from $\RepN$ to $\QCB$.

The category $\QCB$ was first investigated in \cite{Sch:phd} 
under the acronym $\mathsf{AdmSeq}$.
In \cite{Bau00}, $\QCB$ was shown 
to be the same category as the category $\mathsf{PQ}$ of sequential spaces with an $\omega$-projecting cover.
Later the name `qcb-space' was coined by A.~Simpson in \cite{Sim03}. 
Further information about $\QCB$ can be found in \cite{BSS06,ELS04,LN15}.


\subsection{Operators on Qcb-Spaces} 
\label{sub:closureprop:qcb}

The category $\QCB$ of qcb-spaces and its full subcategory $\QCBZ$ have excellent closure properties.

\pagebreak[3]
\begin{theorem}\label{th:qcb:ccc}
 The categories $\QCB$ and $\QCBZ$\index{qcb-space}
 are cartesian closed\index{cartesian closed}
 and have all countable limits and all countable colimits.
\end{theorem}

\noindent
For short this means that we can form within $\QCB$ and $\QCBZ$ countable products, function spaces, subspaces, countable coproducts and quotients.

We sketch how to construct qcb-spaces\index{construction of qcb-spaces} from others.
Unless stated otherwise, the respective operators preserve the $T_0$-property, hence they apply both to $\QCB$ and $\QCBZ$. 

\pagebreak[3]
\begin{example}[Construction of qcb-spaces]\label{ex:construction:qcb} \rm \quad
\begin{enumerate} \vspace*{-0.5ex}
\item \emph{Binary products:}
 The $\QCB$-product of $\WX$ and $\WY$ is the sequential coreflection (see Subsection~\ref{sub:Sequentialisation}) of the topological product $\WX \times_\Top \WY$.
\item \emph{Countable products:}
 Similarly, the $\QCB$-product $\prod_{i\in \IN} \WX_i$ of a sequence $(\WX_i)_{i \in \IN}$ of qcb-spaces
 is the sequential coreflection of the Tychonoff product of these spaces.
 The countable product $\prod_{i\in \IN} \WX$ of a single qcb-space $\WX$ with itself is homeomorphic to the $\QCB$-exponential $\WX^\IN$ (see below),
 where $\IN$ are the natural numbers equipped with the discrete topology.
\item \emph{Exponentiation:}
 Given two qcb-spaces $\WX$ and $\WY$, 
 the carrier set of the $\QCB$-exponential, which we denote by $\WY^\WX$,
 is the set $\setC(\WX,\WY)$ of all total continuous functions $f\colon \WX \to \WY$.
 The topology is given by the sequentialisation of the compact-open topology on $\setC(\WX,\WY)$.
 The compact-open topology is generated by the subbasic opens
 \[
   \langle K,V \rangle:=\big\{ f \in \setC(\WX,\WY) \,\big|\, f[K] \subseteq V \big\} \,,
 \]
 where $K$ varies over the compact subsets of $\WX$ and $V$ over the open subsets of $\WY$.
 For $\WY^\WX$ to be a $T_0$-space it suffices that $\WY$ has the $T_0$-property.
 The convergence relation on $\WY^\WX$ is continuous convergence:
 $(f_n)_n$ \emph{converges continuously to} $f$, if
 $(f_n(x_n))_n$ converges in $\WY$ to $f(x)$, whenever $(x_n)_n$ is a sequence converging in $\WX$ to $x \in \WX$.
\item \emph{Subspaces:}
 Given a qcb-space $\WX$ and a subset $Z$ of $\WX$,
 the subspace topology on $Z$ inherited from $\WX$ is not necessarily sequential.
 Instead one has to consider the \emph{subsequential topology} on $Z$ which is defined as the sequentialisation of subspace topology. This topology is indeed a qcb-topology.
\item \emph{Equalisers:}
 Let $f,g\colon \WX \to \WY$ be continuous functions between $\WX,\WY \in \QCB$.
 The set $E:=\{ x \in \WX \,|\, f(x)=g(x) \}$ endowed with the subsequential topology along with the continuous inclusion map forms an equaliser to $f,g$ in $\QCB$.
\item \emph{Countable limits:}
 Since $\QCB$ and $\QCBZ$ have countable products and equalisers, both have all countable limits. 
 Hence both categories are countably complete.
\item \emph{Meets:}
 Given a sequence of qcb-spaces $(\WX_i)_i$,
 their \emph{meet} (or \emph{conjunction}) $\bigwedge_{i\in \IN} \WX_i$ 
 is defined to have the intersection as its carrier set.
 The topology is the sequentialisation of the meet topology which is generated by the basis
 \[
  \big\{ \bigcap\nolimits_{i=0}^k U_i \,\big|\, k \in \IN,\; \text{$U_i$ open in $\WX_i$ for all $i \leq k$} \big\}
  \,.
 \] 
 Since this is a countable limit construction, $\bigwedge_{i\in \IN} \WX_i$ is a qcb-space.
 The binary meet $\WX_1 \wedge \WX_2$ is defined analogously.
 Note that the Euclidean reals $\IR$ arises as the meet of the lower reals $\IR_<$ and the upper reals $\IR_>$.
\item \emph{Countable coproduct:}
 The countable coproduct of a sequence $(\WX_i)_{i \in \IN}$ of qcb-spaces is constructed as in the category $\Top$ of topological spaces.
 The carrier set is the disjoint union $\bigcup_{i\in \IN} \{ (i,x) \,\big|\, x \in \WX_i\}$.
 A subset $U$ is open if, and only if,
 for every $i \in \IN$ the set $\{ x \in \WX_i \,|\, (i,x) \in U \}$ is open in $\WX_i$.
\item \emph{Quotients:}
 For any qcb-space $\WX$ and any quotient map $q:\WX \twoheadrightarrow \WZ$ the topological space $\WZ$ is a qcb-space as well.
 But in general $\WZ$ does not have the $T_0$-property, even if $\WX$ is $T_0$.
\item \emph{Kolmogorov quotient (``$T_0$-fication''):}
 For a qcb-space $\WX$ the equivalence relation of topological indistinguishability is defined by
 $x \equiv_\WX x' :\Longleftrightarrow \eta_\WX(x)=\eta_\WX(x')$,
 where $\eta_\WX(x)$ denotes the neighbourhood filter $\{ U \text{ open} \,|\, x \in U\}$ of $x$.
 
 The Kolmogorov quotient of a $T_0$-space is just the space itself.
 If $\WX$ is not $T_0$, then
 the \emph{Kolmogorov quotient} $\KQ(\WX)$ of $\WX$ consists of the equivalence classes of $\equiv_\WX$ 
 and its topology is the quotient topology induced by the surjection $[.]_{\equiv_\WX}$ that maps $x$ to its equivalence class.

\item \emph{Coequalisers in $\QCB$:}
 Let $f,g\colon \WA \to \WX$ be continuous functions between $\WA,\WX \in \QCB$.
 We consider the transitive hull $\approx$ of the reflexive relation 
 \[
  \big\{ (x,x) \,\big|\, x \in \WX \big\} 
  \cup 
  \big\{ (f(a),g(a)) \,\big|\, a \in \WA \big\}
  \cup
  \big\{ (g(a),f(a)) \,\big|\, a \in \WA \big\}\,,
 \]
 define $\WZ$ as the space of equivalence classes of $\approx$ 
 and endow $\WZ$ with the quotient topology induced by the surjection $[.]_\approx$ that maps $x$ to its equivalence class.
 Then $\WZ$ and $[.]_\approx$ form a coequaliser to $f,g$ in $\QCB$.
\item \emph{Coequalisers in $\QCBZ$:}
 Given two continuous functions $f,g\colon \WA \to \WX$ between spaces $\WA,\WZ \in \QCBZ$,
 we first construct $\WZ$ and $[.]_\approx$ as above.
 If $\WZ$ is a $T_0$-space, then this pair is also a coequaliser in $\QCBZ$.
 Otherwise we take the Kolmogorov quotient of $\WZ$ 
 together with the composition $[.]_{\equiv_\WZ} \circ [.]_\approx$.
\item \emph{Countable colimits:}
 Since $\QCB$ and $\QCBZ$ have countable coproducts and coequalisers, both have all countable colimits. 
 Hence both categories are countably cocomplete.
\end{enumerate}
\end{example}

More information can be found in \cite{BSS06,ELS04,Sch:phd}.


\subsection{Powerspaces in $\QCBZ$}
\label{sub:powerspaces}

We discuss spaces\index{powerspace} of families of subsets of qcb-spaces like open, closed, overt and compact subsets.


\subsubsection{Open subsets}
\label{subsub:open}

Given a qcb-space $\WX$, we endow the family $\setO(\WX)$ of open sets of $\WX$ with the Scott-topology on the complete lattice $(\setO(\WX),{\subseteq})$ 
and denote the resulting topological space by $\calO(\WX)$.
A subset $H \subseteq \setO(X)$ is called \emph{Scott open},
if $H$ is upwards closed and $D\cap H$ is non-empty
for each directed subset $D \subseteq \setO(\WX)$ with $\bigcup(D)\in H$.
Since any qcb-space is hereditarily Lindel{\"of}, the Scott-topology coincides with the $\omega$-Scott-topology.
The \emph{$\omega$-Scott-topology} consists of all upwards-closed subsets $H$ such that
for all sequences $(U_i)_{i \in \IN}$ with $\bigcup_{i \in \IN} U_i \in H$
there exists some $n$ such that $\bigcup_{i=0}^n U_i \in H$.

The space $\calO(\WX)$ is homeomorphic to the $\QCB$-exponential $\IS^\WX$, where $\IS$ denotes the Sierpi{\'n}ski space (cf.\ Example~\ref{ex:Sierpinski}).
The homeomorphism is given by the natural identification of an open set $U$ with its continuous characteristic function $\cf(U)\colon \WX \to \IS$ defined by $\cf(U)(x)=\top :\Longleftrightarrow x \in U$.

\pagebreak[3]
\begin{proposition}\label{p:calO} 
 Let $\WX$ be a qcb-space.
\begin{enumerate} 
 \item \label{i:calO:qcb}
  The space $\calO(\WX)$ is a qcb$_0$-space.
 \item \label{i:calO:homeomorphic}
  The space $\calO(\WX)$ is homeomorphic to $\IS^\WX$ via the map $\cf\colon \setO(\WX) \to \IS^\WX$.
 \item
  The Scott-topology on $\setO(\WX)$ is the sequentialisation of the \emph{compact-open} topology on $\setO(\WX)$, 
  which is generated by basic sets $K^{\subseteq}:=\{ U \in \setO(\WX)\,|\, K \subseteq U \}$,
  where $K$ varies over the compact subsets of $\WX$.
 \item \label{i:calO:cup:cap}
  A sequence $(U_n)_n$ converges to $V$ in $\calO(\WX)$
  if, and only if,
  \[
    W_k:=\bigcap\limits_{n \geq k} U_n \cap V \text{ is open for all $k \in \IN$}
    \;\;\text{and}\;\;
    V=\bigcup\limits_{n \in \IN} W_n.
  \]
 \item
  Binary intersection $\cap\colon \calO(\WX) \times \calO(\WX) \to \calO(\WX)$
  is sequentially continuous.
 \item
  Countable union $\bigcup\colon \prod_{i\in\IN} \calO(\WX) \to \calO(\WX)$ is sequentially continuous.
\end{enumerate}
\end{proposition}

A space such that the Scott-topology on its lattice of open subsets agrees with the compact-open topology is called \emph{consonant}. 
M.~de Brecht has shown that quasi-Polish spaces are consonant (cf.\ \cite{dBSS16}),
whereas the space of rational numbers is known to be dissonant (= non-consonant).

Importantly, any qcb$_0$-space $\WX$ embeds sequentially into the space $\calO(\calO(\WX))$ via the neighbourhood filter map $\eta_\WX$.
Recall that this means that $\eta_\WX$ is injective and $(x_n)_n$ converges to $x$ in $\WX$ if, and only if, $(\eta_\WX(x_n))_n$ converges to $\eta_\WX(x)$ in $\calO(\calO(\WX))$.

\pagebreak[3]
\begin{proposition}
\label{p:X:embeds:into:OOX}
 Let $\WX$ be a qcb$_0$-space.
\begin{enumerate}
  \item 
   For every $x\in \WX$, the open neighbourhood filter 
   $\eta_\WX(x):=\{ U \text{ open} \,|\, x \in U \}$ is open in $\calO(\WX)$.
  \item
   The map $\eta_\WX\colon \WX \to \calO(\calO(\WX))$ is a sequential embedding of $\WX$ into $\calO(\calO(\WX))$.
  \item
   The map $\eta_\WX \colon \WX \to \IS^{\IS^\WX}$ defined by $\eta_\WX(x)(h) := h(x)$ is a sequential embedding of $\WX$ into the $\QCB$-exponential $\IS^{\IS^\WX}$.
\end{enumerate}
\end{proposition}

\noindent
Recall that $\calO(\calO(\WX))$ is homeomorphic to $\IS^{\IS^\WX}$.
We remark that $\eta_\WX$ is even a topological embedding of $\WX$ into $\IS^{\IS^\WX}$, 
meaning that $\eta_\WX$ forms an homeomorphism between $\WX$ and its image endowed with the subspace topology. 
For more information see \cite{dBSS16,Sch04}.


\subsubsection{Overt subsets}
\label{subsub:overt sets}

On the family of all subsets of $\WX$ we consider the topology of \emph{overtness} which has as subbasis
all sets of the form
$\Diamond U:=\{ M \subseteq \WX \,|\,  M \cap U \neq \emptyset \}$, where $U$ varies over the open sets of $\WX$.
By $\calV(\WX)$ we denote the space all subsets of $\WX$ equipped with the sequentialisation of this topology.

\pagebreak[3]
\begin{proposition}
\label{p:calV:QCB}
 Let $\WX$ be a qcb-space.
\begin{enumerate}
  \item 
   $\calV(\WX)$ is a qcb-space.
  \item
   For every subset $M$, $\eta^\Diamond_\WX(M):= \{ U \text{ open} \,|\, M \in \Diamond U \}$ is open in $\calO(\WX)$.   
  \item
   Binary union on $\calV(\WX)$ is sequentially continuous.
  \item
   Binary intersection on $\calV(\IR)$ is sequentially discontinuous.
  \item
   Two subsets $M_1,M_2 \subseteq \WX$ have the same neighbourhoods in $\calV(\WX)$ if, and only if, their topological closure agrees.
\end{enumerate}
\end{proposition}

Overt subspaces were investigated by P.~Taylor in \cite{Tay10}.


\subsubsection{Closed subsets}
\label{subsub:closed}
The upper Fell topology on the family of closed subsets of a qcb-space $\WX$
is generated by the basic sets $\{ A \text{ closed} \,|\, A \cap K =\emptyset\}$,
where $K$ varies over the compact subsets of $\WX$.
This topology is not sequential, unless $\WX$ is a consonant space.
By $\calAm(\WX)$ we denote the space of closed subspaces equipped with the sequentialisation of the upper Fell topology.
Clearly $\calAm(\WX)$ is homeomorphic to $\calO(\WX)$ via set complementation.
So $\calAm(\WX)$ is a qcb$_0$-space.

The lower Fell topology on $\setA(\WX)$ has as subbasis all sets of the form
$\Diamond U:=\{ A \text{ closed} \,|\,  A \cap U \neq \emptyset \}$.
By $\calAp(\WX)$ we denote the space of closed sets of $\WX$ carrying the sequentialisation of the lower Fell topology on $\WX$.
This space is homeomorphic to the Kolmogorov quotient of the space $\calV(\WX)$ from Subsection~\ref{subsub:overt sets}.

\pagebreak[3]
\begin{proposition}
\label{p:calAp:QCB}
 Let $\WX$ be a qcb-space.
\begin{enumerate}
 \item 
   $\calAp(\WX)$ is a qcb$_0$-space.
 \item
  The function $\eta^\Diamond_\WX\colon \calAp(\WX) \mapsto \calO(\calO(\WX))$ mapping $A$ to $\{ U \text{ open} \,|\, A \in \Diamond U \}$
  is a sequential embedding of $\calAp(\WX)$ into $\calO(\calO(\WX))$.
 \item
  Binary union on $\calAp(\WX)$ is sequentially continuous.
 \item
  Binary intersection on $\calAp(\IR)$ is sequentially discontinuous.
\end{enumerate}
\end{proposition}

Computability on closed subsets of computable metric spaces was studied by V.~Brattka and G.~Presser in \cite{BP03}.
P.~Hertling (\cite{Her02b}) and M.~Ziegler (\cite{Zie04})
investigated computability of regular closed subsets of $\IR^d$.


\subsubsection{Compact subsets}
\label{subsub:compact}
Several notions of compactness are studied in topology:
A subset $K$ of a topological space is \emph{compact}, if every open cover of $K$ has a finite subcover;
$K$ is called \emph{countably compact}, if every countable cover of $K$ has a finite subcover;
$K$ is called \emph{sequentially compact}, if every sequence $(x_n)_n$ in $K$ has a subsequence that converges to some element in $K$.
For subsets of qcb-spaces, however, these three notions coincide (cf.\ \cite{BSS07}).

The \emph{Vietoris topology} on the family of all compact subsets of a topological space~$\WX$ is generated by the subbasic open sets 
\[
  \Box U := \{ K \text{ compact in $\WX$} \,|\, K \subseteq U\} 
  \;\;\text{and}\;\;
 \Diamond U := \{ K \text{ compact in $\WX$} \,|\, K \cap U \neq \emptyset\} \,,
\] 
where $U$ varies over the open subsets of $\WX$. 
By $\KViet(\WX)$ we denote the space of all compact subsets of $\WX$ 
equipped with the sequentialisation of the Vietoris topology.

The \emph{upper Vietoris topology} on the family of compact sets is generated by the basic open sets $\Box U$ only.
Unless $\WX$ is a $T_1$-space, the upper Vietoris topology does not enjoy the $T_0$-property.
Therefore we restrict ourselves to the family of all saturated compact subsets of $\WX$, on which the corresponding upper Vietoris topology is indeed $T_0$.
A subset $M$ is \emph{saturated} if, and only if, it is equal to its \emph{saturation} $\mathalpha{\uparrow}M$,
which is the intersection of all open sets containing $M$.
In a $T_1$-space the saturation of a set yields the set itself.
By $\Kup(X)$ we denote the space of compact saturated subsets of $\WX$ equipped with the sequentialisation of the upper Vietoris topology.

\pagebreak[3]
\begin{proposition}
\label{p:upvietoris:qcb}
 Let $\WX$ be a qcb$_0$-space.
\begin{enumerate}
 \item 
  $\Kup(\WX)$ and $\KViet(\WX)$ are qcb$_0$-spaces.
 \item
  For every compact subset $K$, $\eta^\Box_\WX(K):= \{ U \text{ open} \,|\, K \in \Box U \}$ is open in $\calO(\WX)$.
 \item
  The map $\eta^\Box_\WX$ 
  yields a sequential embedding of $\Kup(\WX)$ into $\calO(\calO(\WX))$.
 \item
  Binary union on $\Kup(\WX)$ and on $\KViet(\WX)$ are sequentially continuous.
\end{enumerate}
\end{proposition}

Note that the intersection of two saturated compact sets need not be compact,
unless $\WX$ is for example Hausdorff. 

The idea to compute on compact subsets via their embedding into $\calO(\calO(\WX))$ was independently developed by M.~Escard{\'o} in \cite{Esc04,Esc08} and M.~Schr\"oder in \cite{Sch:phd}.
V.~Brattka and G.~Presser investigated computability on compacts subsets of computable metric spaces in \cite{BP03}.

%

\subsection{Qcb-Spaces and Admissibility}
\label{sub:qcb:admissible}

Qcb$_0$-spaces\index{qcb-space} form exactly the class of sequential spaces that have an admissible\index{admissible representation} representation.
For the proof, one first observes that any qcb$_0$-space has a quotient representation:
If a $T_0$-space $\WX$ is a topological quotient space of a second-countable space $\WW$, then $\WX$ is also a topological quotient of the Kolmogorov quotient $\KQ(\WW)$ of $\WW$.
By Example~\ref{ex:deltabeta}, $\KQ(\WW)$ has an admissible representation;
its composition with the quotient map $\KQ(\WW) \twoheadrightarrow\WX$ yields a quotient representation for $\WX$.

Now we define an operator that converts a quotient representation\index{quotient representation} 
$\delta$ into an admissible representation\index{admissible representation}.
From Propositions~\ref{p:funcrep:adm}, \ref{p:seqembedding:adm} and~\ref{p:X:embeds:into:OOX} 
we know for a quotient representation $\delta$ of a $T_0$-space $\WX$:

\begin{itemize}
 \item[(a)] 
  $\WX$ embeds sequentially into the $\QCBZ$-exponential $\IS^{\IS^\WX}$ via $\eta_\WX$.
 \item[(b)] 
  $[\delta\to\varrho_\IS]$ is an admissible representation for $\IS^\WX$, even if $\delta$ is not admissible.  
 \item[(c)]
  Corestrictions of admissible representations are admissible again.
\end{itemize}

\noindent
These observations motivate to define an operator by
\[
  \delta^\IS(p)=x \;:\Longleftrightarrow\;
  \eta_\WX(x)=[[\delta\to\varrho_\IS]\to\varrho_\IS](p)
  \,.
\]
The above facts imply that $\delta^\IS$ is admissible.
In the case that that the final topology $\tau_\delta$ of $\delta$ lacks the $T_0$-property, it makes sense to generalise this definition by letting $\delta^\IS$ be the representation of the Kolmogorov quotient $\KQ(X,\tau_\delta)$ given by the equivalence
\[
  \delta^\IS(p)=[x]_\equiv \;:\Longleftrightarrow\;
  \eta_{(X,\tau_\delta)}(x)=[[\delta\to\varrho_\IS]\to\varrho_\IS](p) 
  \,,
\]
where $\equiv$ denotes the equivalence relation of topological indistinguishability (see Example~\ref{ex:construction:qcb}). 
In any case, the map $\delta^\IS$ is indeed an admissible representation. 

\pagebreak[3]
\begin{theorem}[Admissibility of $\delta^\IS$]
\label{th:deltaIS:admissible} \quad\\
\noindent 
 Let $\delta$ be a representation of a set $X$.
\begin{enumerate}
 \item
  $\delta^\IS$ is an admissible representation for $(X,\tau_\delta)$,
  if $\tau_\delta$ has the $T_0$-property.
 \item
  $\delta^\IS$ is an admissible representation for $\KQ(X,\tau_\delta)$.
 \item
  If $\delta$ is a quotient representation for a $T_0$-space $\WX$,
  then $\delta^\IS$ is an admissible representation for $\WX$.
\end{enumerate}  
\end{theorem}  
  
We obtain the following characterisation of admissible representations.  
  
\pagebreak[3]  
\begin{theorem}[Characterisation of admissibility]
\label{th:charac:deltaIS:admissible}
  A representation $\delta$ of a set $X$ is an admissible\index{admissible representation} representation for $(X,\tau_\delta)$, 
  if and only if, 
  $\delta$ is topologically equivalent to $\delta^\IS$ and $\tau_\delta$ satisfies the $T_0$-property.
\end{theorem}

Along with Propositions~\ref{p:adm:quotient} and~\ref{p:adm:pseudobase} we obtain the following characterisation.

\pagebreak[3]
\begin{theorem}[Sequential spaces with admissible representation] \quad
\label{th:qcb:admissible}
\begin{enumerate} \vspace*{-0.5ex}
  \item 
   A sequential topological space has an admissible representation
   if, and only if, it is a qcb$_0$-space.
  \item
   A topological space has an admissible quotient representation
   if, and only if, it is a qcb$_0$-space.
  \item
   A topological space has an admissible representation
   if, and only if, its sequential coreflection is a qcb$_0$-space.
\end{enumerate} 
\end{theorem}

From our previous characterisation in Proposition~\ref{p:adm:pseudobase} we deduce that a sequential space is a qcb-space\index{qcb-space} if, and only if, it has a countable pseudobase.
We remark that there are canonical ways to construct pseudobases for qcb-spaces that are built by the operators considered in Example~\ref{ex:construction:qcb} and Subsection~\ref{sub:powerspaces} (cf. \cite{Sch:phd,Sch02}).
Theorem~\ref{th:qcb:admissible} implies that the category of admissibly represented spaces\index{admissibly represented space} and continuously realizable functions is equivalent to $\QCBZ$ (cf.\ \cite{BSS07}).
Proofs of this subsection can be found in~\cite{Sch:phd}.


\subsection{Effectively Admissible Representations}
\label{sub:effadm}

The basic observation in Theorem~\ref{th:charac:deltaIS:admissible}
that a representation $\delta$ is topologically admissible if, and only if, it is topologically equivalent to $\delta^\IS$ gives rise to call a representation $\delta$ \emph{effectively admissible}\index{effectively admissible},
if $\delta$ is even computably equivalent to $\delta^\IS$.

\pagebreak[3]
\begin{proposition}
\label{p:delta^IS:properties}
 Let $\delta,\gamma$ be representations such that their final topologies are~$T_0$.
\begin{enumerate}
 \item \label{en:preservesComputability}
  Any $(\delta,\gamma)$-computable total function $f$ is $(\delta^\IS,\gamma^\IS)$-computable.
 \item 
  Any $(\delta,\gamma)$-continuous total function $f$ is $(\delta^\IS,\gamma^\IS)$-continuous.
 \item
  The representations $\delta^\IS$ and $\gamma^\IS$ are effectively admissible.
 \item
  The signed-digit representation and the Cauchy representation for $\IR$ are effectively admissible.
\end{enumerate}
\end{proposition}
We remark that Proposition~\ref{p:delta^IS:properties}\eqref{en:preservesComputability} does not hold for partial functions: 
it is even possible that a $(\delta,\gamma)$-computable partial function is not $(\delta^\IS,\gamma^\IS)$-continuous.

\pagebreak[3]
\begin{proposition}
\label{p:effadm:constructionI}
Let $\delta$ and $\gamma$ be effectively admissible representations for qcb$_0$-spaces $\WX$ and $\WY$, respectively.
\begin{enumerate}
 \item
  $\delta \boxtimes \gamma$ is an effectively admissible representation for the $\QCB$-product $\WX \times \WY$.
 \item
  $\delta \boxplus \gamma$ is an effectively admissible representation for the $\QCB$-coproduct $\WX \uplus \WY$. 
 \item
  $[\delta\to\gamma]$ is an effectively admissible representation for the $\QCB$-exponential $\WY^\WX$.
 \item
  $\delta \wedge \gamma$ is an effectively admissible representation for the $\QCB$-meet $\WX \wedge \WY$.
 \item
  Let $\WZ$ be a subsequential subspace of $\WX$. \\
  Then the corestriction $\delta|^\WZ$ is an effectively admissible representation for $\WZ$.
 \item 
  Let $e\colon \WZ \to \WX$ be a sequential embedding of a sequential space $\WZ$ into $\WX$.\\
  Then $e^{-1} \circ \delta$ is an effectively admissible representation for $\WZ$.
\end{enumerate}
\end{proposition}

To construct\index{construction of representations} effectively admissible representations 
for the function space $\YY^\XX$ and for the power spaces over $\WX$ considered in Subsection~\ref{sub:powerspaces},
we only need a quotient representation for $\WX$.

\pagebreak[3]
\begin{proposition}
\label{p:effadm:constructionII}
 Let $\delta$ be a quotient representation for a qcb$_0$-space $\WX$.
\begin{enumerate}
 \item
  Let $\gamma$ be an effectively admissible representation for a qcb$_0$-space $\WY$. \\
  Then $[\delta\to\gamma]$ is an effectively admissible representation for $\WY^\WX$.
 \item
  Let $q\colon \WX \to \WZ$ be a quotient map. 
  Then $(q \circ \delta)^\IS$ is an effectively admissible representation of the Kolmogorov quotient of $\WZ$.  
 \item
  Define $\delta^\calO$ by 
  $\delta^\calO(p)=U \,:\Longleftrightarrow\, \cf(U)= [\delta\to\varrho_\IS](p)$. \\
  Then $\delta^\calO$ is an effectively admissible representation for $\calO(\WX)$.
 \item
  Define $\delta^\calAm$ by 
  $\delta^\calAm(p)=A \,:\Longleftrightarrow\, \cf(X\setminus A)= [\delta\to\varrho_\IS](p)$. \\
  Then $\delta^\calAm$ is an effectively admissible representation for $\calAm(\WX)$.  
 \item
  Define $\delta^\calAp$ by 
  $\delta^\calAp(p)=A \,:\Longleftrightarrow\, \eta^\Diamond_\WX(A)=[[\delta\to\varrho_\IS]\to\varrho_\IS](p)$.\\
  Then $\delta^\calAp$ is an effectively admissible representation for $\calAp(\WX)$.
 \item 
  Let $\WX$ be $T_1$.
  Define $\delta^\Kup$ by
  $\delta^\Kup(p)=K \,:\Longleftrightarrow\,\eta^\Box_\WX(K)=[[\delta\to\varrho_\IS]\to\varrho_\IS](p)$. \\
  Then $\delta^\Kup$ is an effectively admissible representation for $\Kup(\WX)$.
 \item 
  Define $\delta^\KViet$ by
  $\delta^\KViet(p)=K \,:\Longleftrightarrow\, \delta^\Kup(\pi_1(p))=K \,\,\&\,\,\delta^\calAp(\pi_2(p))=\Cls(K)$, 
  where $\Cls(K)$ denotes the topological closure of $K$.
  Then $\delta^\KViet$ is an effectively admissible representation for $\KViet(\WX)$. 
\end{enumerate}
\end{proposition}
\noindent
Proofs of this subsection can be found in \cite{Sch:phd}.


\subsection{Effective Qcb-Spaces}
\label{sub:EffQCB}

An \emph{effective qcb-space}\index{effective qcb-space} is a represented space $\XX=(X,\delta)$ such that $\tau_\delta$ has the $T_0$-property and $\delta$ is an effectively admissible representation for the associated qcb-space $(X,\tau_\delta)$.
The category $\EffQCB$ has effective qcb-spaces as its objects and total computable functions as morphisms.
So $\EffQCB$ is a full subcategory of $\RepEff$.

A similar category was first investigated in \cite{Sch:phd,Sch02b} under the acronym $\mathsf{EffSeq}$, but it comprises also effectively multirepresented qcb-spaces by dropping the $T_0$-condition.
A.~Pauly has investigated effective qcb-spaces under the name \emph{computably admissible represented space} (cf.\ \cite{Pau16}).
However, we prefer the former name, because there exist further classes of represented spaces enjoying other reasonable effectivity properties.
We will discuss this in Subsections~\ref{sub:QN} and~\ref{sub:limitspaces}.

\pagebreak[3]
\begin{example}[Basic topological spaces as effective qcb-spaces]
\label{ex:EffQCB}  \rm \quad\\ \noindent
 The following are effective qcb-spaces.
\begin{enumerate}
 \item 
  $(\Nomega,\mathit{id}_{\Nomega})$.
 \item
  $(\IR, \varrho_\IR)$, where $\varrho_\IR$ is the Cauchy representation\index{Cauchy representation} from Subsection~\ref{sub:CauchyRepReal}.  
 \item
  $(\IR, \sdrep)$, where $\sdrep$ is the signed-digit representation\index{signed-digit representation} from Subsection~\ref{sub:SignedDigitRep}.
 \item
  $(\IN, \varrho_\IN)$ with $\varrho_\IN(p)=p(0)$.
 \item
  $(\IS, \varrho_\IS)$ with $\varrho_\IS(p)=\bot :\Longleftrightarrow p=\w{0}^\omega$
  and $\varrho_\IS(p)=\top :\Longleftrightarrow p \neq \w{0}^\omega$.
\end{enumerate}
\end{example}

Computable topological spaces in the sense of T.~Grubba and K.~Weihrauch are effective qcb-spaces.

\pagebreak[3]
\begin{example}[Computable topological spaces]
\label{ex:CompTopSpaces} \rm \quad\\
\noindent
  A \emph{computable topological space}\index{computable topological space}
  (cf.\ \cite{WG09}) is a pair $(X,\beta)$,
  where $\beta$ is a partial numbering of a base of a $T_0$-topology such that there is a computably enumerable set $J \subseteq \IN^3$ satisfying 
  $\beta_a \cap \beta_b=\Tbigcup\{ \beta_k \,|\, (a,b,k) \in J\} $
  and the domain of $\beta$ is computably enumerable.
  The standard representation $\delta_\beta$ from Example~\ref{ex:deltabeta} defined by
  \[  
    \delta_\beta(p)=x \;:\Longleftrightarrow\;
     \En(p) = \{ i \in \dom(\beta) \,|\, x \in \beta_i \} 
  \]    
  is effectively admissible. 
  So $(X,\delta_\beta)$ is an effective qcb-space.
\end{example}

Computable metric spaces can be viewed as computable topological spaces in the sense of Example~\ref{ex:CompTopSpaces} and thus as effective qcb-spaces.

\pagebreak[3]
\begin{example}[Computable metric space]\label{ex:cms} \rm \quad\\
\noindent
 A \emph{computable metric space (cms)}\index{computable metric space} is a metric space $(X,d)$ 
 together with a total numbering $\alpha$ of a dense subset $A$ such that
 the function $(i,j) \mapsto d(\alpha_i,\alpha_j)$ is $(\varrho_\IN\boxtimes\varrho_\IN;\varrho_\IR)$-computable.
 The latter is equivalent to saying that $(d(\alpha_{\fst(k)},\alpha_{\snd(k)}))_{k}$ is a computable sequence of real numbers, where $\fst,\snd\colon \IN \to \IN$ denote the computable projections of a canonical computable pairing function for $\IN$.
 
 A canonical numbering $\beta$ of a base of the topology generated by the metric $d$ is constructed by
 \[
   \beta_k:=\{ x \in X \,|\, d(x,\alpha_{\fst(k)})<2^{-snd(k)} \} \,,
 \]
 The computability condition on $\alpha$ guarantees that $(X,\beta)$ is a computable topological space
 in the sense of Example~\ref{ex:CompTopSpaces}.
 The \emph{Cauchy representation}\index{Cauchy representation} $\varrho_\alpha\colon \Nomega \parto X$ is defined by
 \[
    \varrho_\alpha(p)=x \;:\Longleftrightarrow\; \forall k \in \IN.\, d(x,\alpha_{p(k)}) \leq 2^{-k}
  \]
 for all $p \in \IN$ and $x \in X$. 
 The Cauchy representation is computably equivalent to the standard representation $\delta_\beta$ from Example~\ref{ex:CompTopSpaces}.
 Hence $(X,\varrho_\alpha)$ is an effective qcb-space.
 
 The real numbers $\IR$ together with the numbering $\nu_\IQ$ from Subsection~\ref{sub:CauchyRepReal} is a prime example of a computable metric space. 
 Another prominent example is the infinite-dimensional separable Hilbert space $\ell_2$ (cf.\ \cite{BD07a}).
 For more examples see \cite{BHW08,Wei00}. 
\end{example}

Propositions~\ref{p:effadm:constructionI} and \ref{p:effadm:constructionII} tell us how to construct effectively admissible qcb-spaces from given ones.

\pagebreak[3]
\begin{proposition}[Construction of effective qcb-spaces]
\label{p:construction:effqcb} \quad\\ \noindent
 Let $\XX=(X,\delta)$ and $\YY=(Y,\gamma)$ be effective qcb-spaces.
 The following are effective qcb-spaces, where $\WX$ and $\WY$ denote the qcb$_0$-spaces associated to $\XX$ and $\YY$, respectively.
\begin{enumerate} 
 \item 
  $\XX \times \YY:= (\WX \times \WY, \delta \boxtimes \gamma)$.
 \item
  $\XX \uplus \YY:= (\{1\} \times \WX \cup \{2\} \times \WY, \delta \boxplus \gamma)$.
 \item
  $\YY^\XX:= (\WY^\WX, [\delta \to \gamma])$.
 \item
  $\XX \wedge \YY:= (\WX \wedge \WY, \delta \wedge \gamma)$.
 \item
  $\calO(\XX):= (\calO(\WX),\delta^\calO)$.
 \item
  $\calAm(\XX):= (\calAm(\WX),\delta^\calAm)$.
 \item
  $\calAp(\XX):= (\calAp(\WX),\delta^\calAp)$.
 \item
  $\Kup(\XX):= (\Kup(\WX),\delta^\Kup)$.
 \item
  $\KViet(\XX):= (\KViet(\WX),\delta^\KViet)$.
 \item
  $\XX|^Z:=(Z,\delta|^Z)$ for $Z \subseteq \XX$, where $\delta|^Z$ is the subspace representation.
\end{enumerate}
\end{proposition}

\noindent
In Proposition~\ref{p:construction:effqcb} we have deliberately used qcb-spaces rather than sets as first component: by Propositions~\ref{p:effadm:constructionI} and~\ref{p:effadm:constructionII}
they are the qcb-spaces asscociated to the represented spaces in (1) to (9).
The qcb-space asscociated to $\XX|^Z$ in (10) carries the subsequential topology on $Z$.
The category $\EffQCB$ has nice closure properties.

\pagebreak[3]
\begin{theorem}
\label{th:effqcb:ccc}
 The category $\EffQCB$ of effective qcb-spaces\index{effective qcb-space}
 and total computable functions
 is cartesian closed\index{cartesian closed}, finitely complete and finitely cocomplete.
\end{theorem}

\subsection{Basic Computable Functions on Effective Qcb-Spaces}
\label{sub:ComputableOperations}

We present some basic computable functions on effective qcb-spaces.
More can be found in A.~Pauly's paper \cite{Pau16}.

\pagebreak[3]
\begin{proposition}[\cite{Pau16}]  
\label{p:effqcb:compfunctions}
 Let $\XX,\YY,\ZZ$ be effective qcb-spaces.
 The following functions are computable:
\begin{enumerate}
 \item
  $\mathit{pr}_1\colon \XX \times \YY \to \XX$,\; $(x,y) \mapsto x$, \;
  $\mathit{pr}_2\colon \XX \times \YY \to \YY$,\; $(x,y) \mapsto y$.
 \item
  $\mathit{eval}\colon \YY^\XX \times \XX \to \YY$, $(f,x) \mapsto f(x)$.
 \item
  $\mathit{curry}\colon \YY^{\ZZ \times \XX} \to (\YY^\XX)^\ZZ$,\; $(\mathit{curry}(f)(z))(x):=f(z,x)$.
 \item
  $\mathit{uncurry}\colon (\YY^\XX)^\ZZ \to \YY^{\ZZ \times \XX}$,\; $\mathit{uncurry}(h)(z,x):=h(z)(x)$.
 \item
  $\mathit{join}\colon \XX^\ZZ \times \YY^\ZZ \to (\XX \times \YY)^\ZZ$,\; $\mathit{join}(f,g)(z):=(f(z),g(z))$.    
 \item 
  ${}^\complement \colon \calO(\XX) \to \calAm(\XX)$,\; $U \mapsto \XX \setminus U$.
 \item 
  ${}^\complement \colon \calAm(\XX) \to \calO(\XX)$,\; $A \mapsto \XX \setminus A$.    
 \item
  $\cap\colon \calO(\XX) \times \calO(\XX) \to \calO(\XX)$,\; $(U,V) \mapsto U \cap V$.
 \item
  $\cup\colon \calO(\XX) \times \calO(\XX) \to \calO(\XX)$,\; $(U,V) \mapsto U \cup V$.
 \item
  $\bigcup\colon \calO(\XX)^\IN \to \calO(\XX)$,\; $(U_n)_n \mapsto \bigcup_{n \in \IN} U_n$.
 \item
  $\cup\colon \Kup(\XX) \times \Kup(\XX) \to \Kup(\XX)$,\; $(K,L) \mapsto K \cup L$.
 \item
  $\cup\colon \calAp(\XX) \times \calAp(\XX) \to \calAp(\XX)$,\; $(A,B) \mapsto A \cup B$.  
 \item
  $\cap\colon \Kup(\XX) \times \calAm(\XX) \to \Kup(\XX)$,\; $(K,A) \mapsto K \cap A$.  
 \item 
  $\times\colon \calO(\XX) \times \calO(\YY) \to \calO(\XX \times \YY)$,\; $(U,V) \mapsto U \times V$.
 \item 
  $\times\colon \calAm(\XX) \times \calAm(\YY) \to \calAm(\XX \times \YY)$,\; $(A,B) \mapsto A \times B$. 
 \item
  $\iota\colon \XX \to \Kup(\XX)$, \;$x \mapsto {\uparrow}\{x\}$.
 \item
  $\iota\colon \XX \to \calAp(\XX)$, \;$x \mapsto \Cls\{x\}$.  
 \item
  $\mathit{preimage}\colon \YY^\XX \times \calO(\YY) \to \calO(\XX)$, \; $(f,V) \mapsto f^{-1}[V]$.
 \item
  $\mathit{image}\colon \YY^\XX \times \Kup(\XX) \to \Kup(\YY)$, \; $(f,K) \mapsto {\uparrow}f[K]$.    
 \item
  $\forall\colon \calO(\XX \times \YY) \times \Kup(\XX) \to \calO(\YY)$,\;
  $(W,K) \mapsto \{ y \in \YY \,|\, \forall x \in K. (x,y) \in W\}$.
 \item
  $\exists\colon \calO(\XX \times \YY) \times \calAp(\XX) \to \calO(\YY)$,\;
  $(W,A) \mapsto \{ y \in \YY \,|\, \exists x \in A. (x,y) \in W\}$. 
\end{enumerate}
\end{proposition}
\noindent
We emphasise that $\forall(W,K)$ and $\exists(W,A)$ are indeed open subsets of the qcb$_0$-space associated to $\YY$.
Related results can be found in \cite{BP03,Esc04,Esc08,Wei00}.


\subsection{Effective Hausdorff Spaces}
\label{sub:EffectiveHausdorff}

Topological spaces $\WX$ in which every converging sequence has a unique limit are called \emph{sequentially Hausdorff}.
This condition is equivalent to sequential continuity of the inequality test ${\neq}\colon \WX \times \WX \to \IS$ on $\WX$ which is defined by ${\neq}(x,y)=\top :\Longleftrightarrow x\neq y$.
Every Hausdorff space is sequentially Hausdorff, sequentially Hausdorff spaces are $T_1$.
This gives rise to the notion of 
an \emph{effective Hausdorff space}\index{effective Hausdorff space}:
this is an effective qcb-space such that its inequality test ${\neq}$ is computable.
A similar notion has been investigated by A.~Pauly in \cite{Pau16}: 
he calls a represented space $\XX$ \emph{computably $T_2$}, if the inequality test is computable; its representation, however, is not required to be effectively admissible. 
Any computable metric space forms an effective Hausdorff space.  

\pagebreak[3]
\begin{proposition}[\cite{Pau16}] 
\label{p:computable:effseqHaus}
 The following functions are computable for an effective Hausdorff space $\XX$:
\begin{enumerate}
 \item
  $\cap\colon \Kup(\XX) \times \Kup(\XX) \to \Kup(\XX)$, \; $(K,L) \mapsto K \cap L$.
 \item
  $\mathit{singleton}\colon \XX \to \calAm(\XX)$, \;$x \mapsto \{x\}$.
 \item
  $\mathit{convert}\colon \Kup(\XX) \to \calAm(\XX)$, \;$K \mapsto K$.
 \item
  $\Graph\colon \XX^\ZZ \to \calAm(\ZZ\times\XX)$,\;
  $f \mapsto \{ (z,f(z)) \,|\, z \in \ZZ\}$ is computable for any effective qcb-space $\ZZ$. 
\end{enumerate}
\end{proposition}  

\noindent
We emphasise that the above functions map indeed into the indicated sets, whenever $\XX$ is an effective Hausdorff space. 
Conversely, if $\XX$ is an effective qcb-space such that one of the functions listed in Proposition~\ref{p:computable:effseqHaus} is well-defined and computable, then the inequality test on $\XX$ is computable, hence $\XX$ is an effective Hausdorff space. This was shown by A.~Pauly in \cite{Pau16}.

The reader should be warned that effective Hausdorff spaces exist which are not Hausdorff in the topological sense, for example the one-point compactification of the Baire space.
Nevertheless, the arguably more precise notion ``effective sequential Hausdorff spaces'' seems too clumsy.
For second-countable spaces, however, sequential Hausdorffness coincide with Hausdorffness.
K.~Weihrauch has investigated and compared several notions of computable Hausdorff separation for second-countable computable topological spaces in \cite{Wei10}.


\subsection{Quasi-Normal Qcb-Spaces}
\label{sub:QN}

We present a subcategory of admissibly represented Hausdorff spaces that is particularly appropriate for investigating computability in Functional Analysis.

Normality\footnote{A $T_1$-space is called \emph{normal}, if any two disjoint closed sets $A,B$ can be separated by two disjoint open sets $U,V$ in the sense that $A \subseteq U$ and $B \subseteq V$ holds.}
has the disadvantage of not being preserved by exponentiation in $\QCB$.
For example, the $\QCB$-exponential $\IR^{(\IR^\IR)}$ is not normal, although $\IR^\IR$ and $\IR$ are even metrisable.
So we introduce a substitute for normality in the realm of qcb-spaces called quasi-normality.
A \emph{quasi-normal}\index{quasi-normal space} space is defined to be the sequential coreflection (see Subsection~\ref{sub:Sequentialisation}) of some normal Hausdorff space.
We write $\QN$ for the full subcategory of $\QCB$ consisting of all quasi-normal qcb-spaces.
In contrast to normal qcb-spaces, $\QN$ has excellent closure properties.

\pagebreak[3]
\begin{theorem}[Closure properties of quasi-normal spaces]
\label{th:QN:closure:properties}\quad\\ \noindent
 The category $\QN$ of quasi-normal qcb-spaces is an exponential ideal
 of $\QCB$ and thus cartesian closed\index{cartesian closed}.
 Moreover, $\QN$ has all countable limits and all countable colimits.
 Countable limits and function spaces are inherited from $\QCB$.
\end{theorem}

\noindent
Hence the $\QCB$-exponential $\IR^{(\IR^\IR)}$ lives inside $\QN$.
Relevant examples for non-metrisable $\QN$-spaces are the space of test functions $\mathcal{D}$ and its dual $\mathcal{D}'$ which is known as the space of distributions (= generalized functions).
Computability on distributions has been investigated by K.~Weihrauch and N.~Zhong in \cite{ZW03}.

The famous Tietze-Urysohn Extension Theorem states that any real-valued function defined on a closed subspace of a normal topological space can be extended to a continuous function defined on the whole space.
Quasi-normal qcb-spaces admit continuous extendability of continuous functions defined on subspaces that are functionally closed. 
\emph{Functionally closed} subsets are exactly the preimages of $0$ under continuous real-valued functions. 
We state a uniform version.

\pagebreak[3]
\begin{theorem}[Uniform Extension Theorem for $\QN$-spaces]
\label{th:QN:uniform:Extension} \quad\\ \noindent
 Let $\WX$ be a functionally closed subspace of a quasi-normal qcb-space $\WY$.
 Then there is a continuous functional $E\colon \IR^\WX \to \IR^\WY$
 satisfying $E(f)(x)=f(x)$ for all continuous functions
 $f\colon \WX \to \IR$ and all $x \in \WX$.
\end{theorem}

A qcb-space $\WX$ is quasi-normal if, and only, if $\WX$ embeds sequentially into $\IR^{(\IR^\WX)}$ via the embedding $\eta_{\WX,\IR} \colon x \mapsto (h \mapsto h(x))$. 
Hence $\IR$ plays a similar role for quasi-normal spaces as the Sierpi{\'n}ski space does for sequential spaces.
Generalising the ideas of Subsection~\ref{sub:EffQCB}, we call a represented space $\XX=(X,\delta)$ 
\emph{effectively quasi-normal}\index{effectively quasi-normal space}, 
if $\tau_\delta$ is a $T_0$-topology and $\delta$ is computably equivalent to the representation $\delta^\IR$ defined by
\[
  \delta^\IR(p)=x \;:\Longleftrightarrow\; \eta_{(X,\tau_\delta),\IR}(x)=[[\delta\to\varrho_\IR]\to\varrho_\IR](p) \,.
\]
The operator $\delta \mapsto \delta^\IR$ behaves like the operator $\delta \mapsto \delta^\IS$ from Subsection~\ref{sub:qcb:admissible}.
\pagebreak[3]
\begin{proposition}
\label{p:delta^IR:properties}
 Let $\delta$ and $\gamma$ be quotient representations for $\QN$-spaces $\WX$ and $\WY$, respectively.
\begin{enumerate}
 \item
  Any $(\delta,\gamma)$-computable total function is $(\delta^\IR,\gamma^\IR)$-computable.
 \item 
  Any $(\delta,\gamma)$-continuous total function is $(\delta^\IR,\gamma^\IR)$-continuous.
 \item
  The spaces $(\WX,\delta^\IR)$ and $(\WY,\gamma^\IR)$ are effectively quasi-normal.
 \item 
  Any effective quasi-normal space is an effective qcb-space.  
 \item 
  Any computable metric space\index{computable metric space} equipped with its Cauchy representation 
  gives rise to an effective quasi-normal space.
 \end{enumerate}
\end{proposition}
\noindent
More information can be found in \cite{Sch09b}.


\subsection{Co-Polish Spaces}
\label{sub:coPolish}

An interesting subclass of quasi-normal spaces are \index{co-Polish space}\emph{co-Polish} Hausdorff spaces \cite{dBPSch:ovchoice}.
These are those Hausdorff qcb-spaces which have an admissible representation with a locally compact domain.
M.~de Brecht has proven that co-Polish Hausdorff spaces are exactly those Hausdorff qcb-spaces $\WX$ for which $\calO(\WX)$, the lattice of open subsets equipped with the Scott topology (cf.\ Subsection~\ref{subsub:open}), is quasi-Polish \cite{dBre13}. 
Moreover, whenever $\WX$ is quasi-normal, the $\QCB$-exponential $\IR^\WX$ is Polish if, and only if, $\WX$ is co-Polish.
Both facts motivate the name ``co-Polish''.

Any locally compact separable metrisable space is co-Polish (cf.\ \cite{Sch04}).
D.~Kunkle showed that any Silva space is co-Polish (cf.\ \cite{KS05}).
Silva spaces (cf.\ \cite{Sil55}) form an important subclass of locally convex vector spaces.


\section{Relationship to Other Relevant Categories}
\label{sec:relationship:QCB}

There are several cartesian closed\index{cartesian closed} categories relevant to Computable Analysis into which $\QCBZ$\index{qcb-space} embeds as subcategory that inherits binary products and function spaces.
As examples we discuss represented spaces, equilogical spaces, limit spaces, filter spaces, 
sequential spaces, compactly generated spaces, and core compactly generated spaces.


\subsection{Represented Spaces}
\label{sub:RepN:QCBZ}

Assuming a mild axiom of choice, $\QCBZ$ embeds into $\RepN$ as a subccc (meaning that $\QCBZ$ inherits the cartesian closed structure).
The inclusion functor maps a qcb$_0$-space $\WX$ to $(\WX,\delta_\WX)$, where $\delta_\WX$ is a chosen admissible representation for $\WX$.
This functor has a left adjoint (cf.\ \cite{MS02}): 
it maps a \index{represented space}represented space $(Y,\gamma) \in \RepN$ to the Kolmogorov quotient of the qcb-space $(Y,\tau_\gamma)$, where $\tau_\gamma$ denotes the final topology.
Similarly, $\EffQCB$ is a full reflective subcategory of $\RepEff$;
the reflection functor maps $(Y,\gamma)$ to the effective qcb-space\index{effective qcb-space} $(\KQ(Y,\tau_\gamma),\gamma^\IS)$.

\subsection{Equilogical Spaces}
\label{sub:equilogical}

The category $\Top$ of topological spaces is well-known not to be cartesian closed.
For example, there exists no exponential to the metrisable space $\IR^\IR$ and $\IR$.
Note that the $\QCB$-exponential $\IR^{(\IR^\IR)}$ does not form an exponential in the category $\Top$, because the evaluation function is not topologically continuous (albeit sequentially continuous).
However, there exist cartesian closed supercategories of $\Top$.
We mention M.~Hyland's category of filter spaces (cf.\ \cite{Heck98,Hyl79}) 
and D.~Scott's category $\Equ$ of equilogical spaces (cf.\ \cite{BBS04}) as important examples.

An \emph{equilogical space} is a topological $T_0$-space together with an equivalence relation on that space. 
It is called \emph{countably-based}, if the topological space has a countable base.
An \emph{equivariant} map between two equilogical spaces is a function between the induced equivalence classes of the equilogical spaces that is realized by a continuous function between the underlying topological spaces.
The category $\Equ$ of equilogical spaces and equivariant maps as well as its full subcategory 
$\oEqu$ of countably-based equilogical spaces are locally cartesian closed \cite{Bau00,BBS04}.

Equilogical spaces form a generalisation of represented spaces.
Any represented space $(X,\delta)$ can be seen as a countably-based equilogical space: 
the underlying topological space is just the domain of the representation and the equivalence relation is given by $p\equiv q :\Longleftrightarrow \delta(p)=\delta(q)$.
The ensuing inclusion functor $\RepN \hookrightarrow \oEqu$ has the disadvantage of not preserving function spaces (cf.\ \cite{Bau00,Bau01}).
However, composition with the functor $\QCBZ \hookrightarrow \RepN$ from Subsection~\ref{sub:RepN:QCBZ} yields an inclusion functor $\QCBZ \hookrightarrow \oEqu$ that does preserve binary products and function spaces. 
So $\QCBZ$ lives inside $\oEqu$ as a subccc.
M.~Menni and A.~Simpson have exhibited $\QCBZ$ as the largest common subcategory of $\oEqu$ and $\Top$ that is cartesian closed and contains all $T_0$-spaces with a countable base (cf.\ \cite{MS02}).


\subsection{Limit Spaces}
\label{sub:limitspaces}

A reasonable notion of admissible representation\index{admissible representation} exists for other classes of spaces besides the classes of the topological spaces discussed before.
We concentrate on limit spaces as the most useful example.
However, the largest class of spaces which can be equipped with an admissible multirepresentation is the class of weak limit spaces with a countable limit base (cf.\ \cite{Sch08a}).
The corresponding category $\oWLim$ is equivalent to the category $\oPFil$ of proper filter spaces with a countable basis studied by M.~Hyland in \cite{Hyl79}. 

A \emph{limit space}\index{limit space} (\cite{Kur66}) equips a set $X$ with a \emph{convergence relation} $\To{}$ between sequences $(x_n)_n$ and points $x$ of $X$.
If $(x_n)_n \To{} x$, then one says that $(x_n)_n$ converges to $x$ in the space $(X,\To{})$ and that $x$ is a limit of the sequence $(x_n)_n$.
The convergence relation of a limit space is required to satisfy the following properties:
\begin{enumerate}
 \item
  any constant sequence $(x)_n$ converges to $x$;
 \item
  any subsequence of a sequence converging to $x$ converges to $x$;
 \item
  a sequence converges to $x$, whenever any of its subsequences has a subsequence converging to $x$.
\end{enumerate} 
Clearly, the convergence relation of a topological space satisfies these axioms.

A partial function between limit spaces is called \emph{continuous}, if it preserves convergent sequences (cf.\ Subsection~\ref{sub:SequentialSpaces}).
The category $\Lim$ of limit spaces and total continuous functions is locally cartesian closed, in contrast to the category $\Seq$ of sequential spaces, which is merely cartesian closed.
Like $\Seq$, it is countably complete and countably cocomplete. 
The category $\Seq$ is a full reflective subcategory of $\Lim$  (cf.\ \cite{Hyl79}).
The inclusion functor maps a sequential space $\WY$ to the limit space that has the same convergence relation.
The reflection functor sends a limit space $\WX$ to the sequential space that carries the topology of \emph{sequentially open} subsets of $\WX$; 
these are exactly the complements of subsets closed under forming limits of converging sequences (cf.\ Subsection~\ref{sub:SequentialSpaces}).
If the resulting sequential space induces the same convergence relation as $\WX$, then $\WX$ is called \emph{topological}.

\pagebreak[3]
\begin{example}[A limit space that is not topological]
\label{ex:nontopological} \rm \quad\\ \noindent
 We equip the set $Z:=\{(1,0)\} \cup (\{2\}\times\IN) \cup (\{3\}\times\IN )$ with the following convergence relation $\To{}$:
 \begin{align*}
   (z_n)_n \to (1,0)
  &\;:\Longleftrightarrow\;
     \forall m. \exists \tilde{n}. \forall n \geq \tilde{n}.\,  z_n \in \{(1,0), (2,j) \,|\, j \geq m \}
  \\
   (z_n)_n \to (2,a)
  &\;:\Longleftrightarrow\;
    \forall m. \exists \tilde{n}. \forall n \geq \tilde{n}.\,   z_n \in \{ (2,a), (3,j) \,|\, j \geq m \}
  \\
   (z_n)_n \to (3,a)
  &\;:\Longleftrightarrow\;
    \exists \bar{n}. \forall n \geq \tilde{n}.\, z_n= (3,a)
 \end{align*} 
 The resulting space $\WZ$ is indeed a limit space. 
 Any sequentially open subset $V$ of $\WZ$ with $(1,0) \in V$ contains the sequence $(3,n)_n$ eventually.
 Therefore the reflection functor maps $\WZ$ to a sequential space in which $(3,n)_n$ converges to $(1,0)$.
 Hence $\WZ$ is not topological. 
\end{example}

A multirepresentation\index{multirepresentation} $\delta$ for a limit space $\WX$ is called \emph{continuous},
if for any sequence $(p_n,x_n)_n$ in $\Graph(\delta)$ such that $(p_n)_n$ converges to $p_0$ in $\Nomega$ the sequence $(x_n)_n$ converges to $x_0$ in $\WX$.
It is called \emph{admissible}\index{admissible representation} for $\WX$, 
if it is continuous and every continuous multirepresentation $\phi$ of $\WX$ is continuously translatable into $\delta$, meaning that the identity function is $(\phi,\delta)$-continuous.
If two elements $x \neq x'$ have a common name under a continuous multirepresentation, then both elements converge to each other as constant sequences.
An admissible multirepresentation for $\WX$ is single-valued (i.e. an ordinary representation) if, and only if, no such elements $x \neq x'$ exist in $\WX$.
This notion of admissibility extends topological admissibility from Definition~\ref{def:admissible:Top} in the sense that a representation is topologically admissible for a qcb$_0$-space $\WY$ if, and only if, it is admissible for $\WY$ viewed as a limit space.

\pagebreak[3]
\begin{example}[$\Lim$-admissible representation] \rm \quad\\ \noindent
 An admissible representation $\zeta$ for the limit space $\WZ$ in Example~\ref{ex:nontopological} is given by
 $\zeta\big( \w{0}^\omega \big):=(1,0),\,
  \zeta\big( \w{0}\w{0}^a\w{1}^\omega \big):=(2,a)=:\zeta\big( (a+1)\w{0}^\omega  \big),\,
  \zeta\big( (a+1) \w{0}^b\w{1}^\omega \big) := (3,b)
 $
 for all $a,b \in \IN$. 
\end{example}

Like $\Lim$, the full subcategory $\oLim$ of limit spaces with an admissible multirepresentation is locally cartesian closed, countably complete and countably cocomplete.
The category $\QCB$ is a full reflective subcategory of $\oLim$; the corresponding reflection functor is just the restriction of the aforementioned functor from $\Lim$ to $\Seq$.
On the other hand, $\oLim$ is a full reflective subcategory of M.~Hyland's category $\oPFil$ of countably-based proper filter spaces.

We now discuss the $\Lambda$-space which plays the role of the Sierpi{\'n}ski space in the realm of limit spaces.

\pagebreak[3]
\begin{example}[The $\Lambda$-space] \rm \quad\\ \noindent
 The \emph{$\Lambda$-space} $\IL$ has as its underlying set $\{\bot,\top,{\uparrow}\}$.
 Its convergence relation is defined by: $(b_n)_n \To{\IL} b_\infty$ iff $b_\infty \neq \top$ or
 $b_n \neq \bot$ for almost all $n$.
 Hence the restriction of $\IL$ to $\{\bot,\top\}$ is the limit space version of the Sierpi{\'n}ski space $\IS$.
 An admissible multirepresentation $\varrho_\IL$ for $\IL$ is defined by
 \[
  \varrho_\IL[p]:=\left\{
  \begin{array}{cl}
   \{ \bot,{\uparrow} \} & \text{if $p=\w{0}^\omega$}
   \\
   \{ \top,{\uparrow} \} & \text{otherwise}
  \end{array} \right.
 \]
 for all $p \in \Nomega$.
 The $\Lambda$-space has the property that any limit space embeds into $\IL^{\IL^\WX}$ via the map $\eta_{\WX,\IL}\colon x \mapsto (h \mapsto h(x))$.
 Remember that any sequential space $\WY$ embeds sequentially into $\IS^{\IS^\WY}$ via an analogous map (cf.\ Proposition~\ref{p:X:embeds:into:OOX}).  
\end{example}

The $\Lambda$-space can be employed to introduce the notion of an \emph{effectively $\Lim$-admissible (multi-)representation} for a limit space.
The definition uses ideas similar to the ones in Subsection~\ref{sub:effadm}.
If $\delta$ lifts convergent sequences of $\WX$, then $\delta^\IL$ defined by
\[
 \delta^\IL[p] \ni x \;:\Longleftrightarrow\; \eta_{\WX,\IL}(x) \in [[\delta\to\varrho_\IL]\to \varrho_\IL][p]
\]
is an admissible multirepresentation for $\WX$.
A multirepresented space $\XX=(X,\delta)$ is called an \emph{effective limit space}\index{effective limit space},
if $\delta$ is computably equivalent to $\delta^\IL$ (cf.\ \cite{Sch:phd,Sch02b}).

\pagebreak[3]
\begin{example}[Effective limit space] \rm \quad\\ \noindent
 For effective limit spaces $\XX,\YY$, the space $\PP(\XX,\YY)$ of partial continuous functions from $\XX$ to $\YY$ from Example~\ref{ex:spacePP} is an effective limit space.
 But in general it is not an effective qcb-space (even if one allowed multirepresentations in the notion of an effective qcb-space).
\end{example}

A further advantage of effective limit spaces over effective qcb-spaces is that the operator $\delta \mapsto \delta^\IL$ preserves computability even of partial functions.

\pagebreak[3]
\begin{proposition}
\label{p:delta^IL:properties}
 Let $\delta$ and $\gamma$ be multirepresentations for limit spaces $\WX,\WY$ lifting converging sequences.
\begin{enumerate}
 \item
  Any $(\delta,\gamma)$-computable partial function is $(\delta^\IL,\gamma^\IL)$-computable.
 \item 
  Any $(\delta,\gamma)$-continuous partial function is $(\delta^\IL,\gamma^\IL)$-continuous.
 \item
  The multirepresentations $\delta^\IL$ and $\gamma^\IL$ are effectively $\Lim$-admissible.
 \item 
  Any effective qcb-space constitutes an effective limit space.  
\end{enumerate}
\end{proposition}

The category $\EffLim$ of effective limit spaces and total computable functions is locally cartesian closed and has finite limits and finite colimits.
It contains $\EffQCB$ as subccc.

More information can be found in \cite{Hyl79,LN15,MS02,Sch:phd,Sch02b}.


\subsection{Cartesian Closed Subcategories of Topological Spaces}
\label{sub:cccTop}

Although the category $\Top$ of topological spaces\index{topological space} is not cartesian closed,
it contains various subcategories which are cartesian closed.
Several of them were investigated by M.~Escard\'{o}, J.~Lawson and A.~Simpson in~\cite{ELS04}.
Besides the category $\Seq$ of \index{sequential space}sequential spaces, we mention the categories $\kTop$ of compactly generated spaces and $\Ccg$ of core compactly generated spaces.

A topological space $\WZ$ is called \emph{compactly generated} (or a \emph{Kelley space}), if any subset $U$ is open in $\WZ$, whenever $p^{-1}[U]$ is open in $\WK$ for every compact Hausdorff space $\WK$ and every continuous function $p\colon \WK \to \WZ$.
For example, any directed-complete poset endowed with the Scott topology is compactly generated.
Compactly generated Hausdorff spaces are known as \emph{k-spaces}
and play an important role in algebraic topology.

Even larger is the category $\Ccg$ of core compactly generated spaces.
\emph{Core compactly generated spaces} arise as all topological quotients of core compact (= exponentiable) topological spaces.
None of the inclusion functors of these categories into $\Top$ preserves finite products. 
The category $\QCB$ of qcb-spaces lives inside $\Seq$, $\kTop$ and $\Ccg$ as a subccc. 
In fact, $\QCB$, $\Seq$, $\kTop$ and $\Ccg$ form an increasing chain of cartesian closed categories such that 
any smaller category inherits binary products and functions spaces from any larger one.
Details can be found in~\cite{ELS04}.

\subsubsection*{Acknowledgement}
I thank Vasco Brattka, Matthew de Brecht, Peter Hertling and Thomas Strei\-cher for valuable discussions.




\section{Attachment}

We present proofs or references to proofs for most of the preceeding results.
This attachment will not appear in the handbook version of this survey.

\subsubsection*{Proofs for Section 2}

\begin{proofof}{Proposition~\ref{p:composition:computable}}
 See \cite[Theorems 2.1.12 and 2.1.13]{Wei00} or \cite{Wei87}.
\end{proofof}

\begin{proofof}{Example~\ref{ex:comp:realfunctions}}
 See \cite[Section 4.3]{Wei00} or use N.~M\"uller's $\mathsf{iRRAM}$ \cite{Mue01}.
\end{proofof}

\begin{proofof}{Proposition~\ref{p:translatable}}
 See \cite[Theorem 3.1.8]{Wei00} or \cite{Wei87}.
\end{proofof}

\subsubsection*{Proofs for Section 3}

\begin{proofof}{Proposition~\ref{p:contrel2topcontinuous}}
 This follows from \cite[Lemma 2.2.3 and Theorem 2.2.10(b)]{Sch:phd}.
 For a direct proof, let $g$ be a continuous $(\delta_\XX,\delta_\YY)$--realiser for $f$ 
 and let $V \in \tau_{\delta_\YY}$.
 Since $\delta_\XX^{-1}[f^{-1}[V]]$ is open in $\dom(\delta_\XX)$ 
 by being equal to the preimage $g^{-1}[\delta_\YY^{-1}[V]]$,
 $f^{-1}[V]$ is open w.r.t.\ the final topology of $\delta_\XX$.
 So $f$ is topologically continuous.
\end{proofof}

\begin{proofof}{Proposition~\ref{p:comprel2topcontinuous}}
 See Propositions~\ref{p:ComputableImpliesContinuous}
 and~\ref{p:contrel2topcontinuous}.
\end{proofof}

\begin{proofof}{Example~\ref{ex:finaltopology:sequential}}
 Sequential spaces are closed under forming topological quotients,
 cf.\ \cite[Exercise 2.4.G]{Eng89}.
\end{proofof}

\begin{proofof}{Proposition~\ref{p:adm:leqtmaximal}}
 See \cite[Proposition 2.3.4]{Sch:phd}. 
\end{proofof}

\begin{proofof}{Theorem~\ref{th:adm:relcont=cont}}
 See \cite[Theorem 4]{Sch02} or \cite[Theorem 2.3.18]{Sch:phd}.
\end{proofof}

\begin{proofof}{Theorem~\ref{th:maintheo}}
 See \cite[Propositions 2.4.34 and 2.3.2]{Sch:phd}.
\end{proofof}

\begin{proofof}{Theorem~\ref{th:maintheo:multivariate}}
 See \cite[Propositions 2.4.32 and 2.3.2]{Sch:phd}. 
\end{proofof}

\begin{proofof}{Proposition~\ref{p:adm:quotient}}
\begin{enumerate}
  \item See \cite[Proposition 2.3.2]{Sch:phd}.
  \item By \cite[Exercise 2.4.G]{Eng89} and \cite[Proposition 2.3.2]{Sch:phd}.
  \item By (2) and \cite[Lemma 2.2.7(5)]{Sch:phd} or \cite[Lemma 6]{Sch02}.
\end{enumerate}
\end{proofof}

\begin{proofof}{Example~\ref{ex:sdrep:adm}}
 The Cauchy representation for $\IR$ is admissible by \cite[Lemma 4.1.6 and 4.1.19]{Wei00}. 
 The signed-digit representation is admissible by \cite[Example 2.3.8]{Sch:phd}.
 The decimal representation is not admissible by \cite[Theorem 4.1.13(6)]{Wei00}.
 This fact was first observed by C.~Kreitz and K.~Weihrauch in \cite[Section~4]{KW85}.
\end{proofof}

\begin{proofof}{Proposition~\ref{p:prodrep:adm}}
 See \cite[Propositions 4.1.8 and 4.5.10]{Sch:phd}.
 For binary products see also \cite[Section 4.3]{Sch02}.
\end{proofof}

\begin{proofof}{Proposition~\ref{p:funcrep:adm}}
\begin{enumerate}
  \item See \cite[Proposition 4.2.5(4)]{Sch:phd}.
  \item See \cite[Proposition 4.2.5(3)]{Sch:phd}.
\end{enumerate}
\end{proofof}

\begin{proofof}{Proposition~\ref{p:seqembedding:adm}}
\begin{enumerate}
  \item This follows from (2). 
  \item See \cite[Proposition 4.1.6]{Sch:phd}.
\end{enumerate}
\end{proofof}

\begin{proofof}{Example~\ref{ex:adm2psb}}
 See \cite[Lemma 11]{Sch02}.
\end{proofof}

\begin{proofof}{Proposition~\ref{p:adm:pseudobase}}
\begin{enumerate}
  \item See \cite[Proposition 13]{Sch02} or \cite[3.1.15]{Sch:phd}.
  \item By (1) and Proposition~\ref{p:adm:quotient}(2).
\end{enumerate}
\end{proofof}

\subsubsection*{Proofs for Section 4}

\begin{proofof}{Theorem~\ref{th:qcb:ccc}}
 See \cite[Theorems 3.2.4, 4.2.6 and Sections 4.1.6, 4.5.5]{Sch:phd},
 where $\QCB$ is denoted by $\mathsf{AdmSeq}$,
 or \cite[Corollary 7.3 and Remark 7.4]{ELS04}
 or:
\begin{enumerate}
  \item Cartesian closedness: see Example~\ref{ex:construction:qcb}(1),(3).
  \item Countable completeness: see Example~\ref{ex:construction:qcb}(6).
  \item Countable co-completeness: see Example~\ref{ex:construction:qcb}(13).
\end{enumerate}
\end{proofof}

\begin{proofof}{Example~\ref{ex:construction:qcb}}
 See \cite[Sections 4.1, 4.2, 4.5]{Sch:phd} or \cite[Section 7]{ELS04}.
\end{proofof}

\noindent
Many of the following proofs are based on the fact that the following multivariate functions 
on the Sierpi{\'n}ski space~$\IS$ (see Example~\ref{ex:Sierpinski}(1))
\begin{enumerate}
  \item 
   $\wedge\colon \IS \times \IS \to \IS$ 
   defined by $\wedge(a,b)=\top :\Longleftrightarrow a=b=\top$,
  \item 
   $\vee\colon\IS \times \IS \to \IS$ 
   defined by $\vee(a,b)=\top :\Longleftrightarrow (a=\top \text{ or } b= \top)$,   
  \item
  $\bigvee\colon \IS^\IN \to \IS$
  defined by 
  $\bigvee((b_n)_n)=\top :\Longleftrightarrow \exists i\in \IN.\, b_i=\top$.
\end{enumerate}
are topologically continuous and computable with respect to the standard admissible representation $\varrho_\IS$ for $\IS$.

\medskip

\begin{proofof}{Proposition~\ref{p:calO}}
 For (1)--(4) see \cite[Proposition 4.1]{dBSS16}.
 Items (5) and (6) follow from the facts that 
 $\wedge\colon \IS \times \IS \to \IS$ and $\bigvee\colon \IS^\IN \to \IS$
 are continuous (even computable w.r.t.\ $\varrho_\IS$)
 and that $\QCB$ is cartesian closed.
\end{proofof}

\begin{proofof}{Proposition~\ref{p:X:embeds:into:OOX}}
 See \cite[Lemma 4.4.8(2) and Proposition 4.3.1]{Sch:phd}.
\end{proofof}

\begin{proofof}{Propositions~\ref{p:calV:QCB} and~\ref{p:calAp:QCB}}
\begin{enumerate}
 \item See \cite[Proposition 4.4.5]{Sch:phd}.
 \item See \cite[Lemma 4.4.4(1)]{Sch:phd}.
 \item 
  Because $\vee\colon \IS \times \IS \to \IS$ is continuous and $\QCB$ is cartesian closed.
 \item 
  Define $A_n:=\{2^{-n},1\}$, $B_n:=\{-2^{-n},1\}$, $A_\infty=B_\infty:=\{0,1\}$.
  Then $(A_n)_n$ converges to $A_\infty$, $(B_n)_n$ converges to $B_\infty$
  in $\calV(\WX)$ and in $\calAp(\WX)$,
  but $(A_n \cap B_n)_n$ does not converge to $A_\infty \cap B_\infty$
  in  $\calV(\WX)$ nor in $\calAp(\WX)$.
 \item
  Because for every subset $M$ and every open set $U$ 
  we have $M \cap U \neq \emptyset \Longleftrightarrow \Cls(M) \cap U \neq \emptyset$.
\end{enumerate}
\end{proofof}

\begin{proofof}{Proposition~\ref{p:upvietoris:qcb}}
\begin{enumerate}
  \item 
   This follows from (3), Proposition~\ref{p:calV:QCB}(1) and Example~\ref{ex:construction:qcb}.
  \item See \cite[Lemma 4.4.8(2)]{Sch:phd}.
  \item See \cite[Proposition 4.4.9(2)]{Sch:phd}.
  \item 
   Because $\eta^\Box_\WX(K_1 \cup K_2)(U)=\eta^\Box_\WX(K_1)(U) \wedge \eta^\Box_\WX(K_2)(U)$,
   $\wedge\colon \IS \times \IS \to \IS$ is continuous, 
   $\QCB$ is cartesian closed
   and binary union on $\calV(\WX)$ is sequentially continuous.
\end{enumerate}
\end{proofof}

\begin{proofof}{Theorem~\ref{th:deltaIS:admissible}}
 See \cite[Proposition 4.3.2(3) and Lemma 2.4.17]{Sch:phd}.  
\end{proofof}

\begin{proofof}{Theorem~\ref{th:charac:deltaIS:admissible}}
 See Theorem~\ref{th:deltaIS:admissible} and Proposition~\ref{p:translatable}(5).
\end{proofof}

\begin{proofof}{Theorem~\ref{th:qcb:admissible}}
\begin{enumerate}
  \item See \cite[Theorem 3.2.4(2)]{Sch:phd}.
  \item See (1) and Proposition~\ref{p:adm:quotient}(2). 
  \item See (1) and Proposition~\ref{p:adm:quotient}(3). 
\end{enumerate}
\end{proofof}

\begin{proofof}{Proposition~\ref{p:delta^IS:properties}}
\begin{enumerate}
  \item See \cite[Proposition 4.3.5(3)]{Sch:phd}.
  \item See \cite[Proposition 4.3.5(2)]{Sch:phd}.  
  \item See \cite[Lemma 4.3.6(2)]{Sch:phd}.
  \item 
   This follows from \cite[Lemma 4.1.6, Theorem 7.2.5]{Wei00},         
   \cite[Proposition 4.3.16]{Sch:phd} and Proposition~\ref{p:translatable}(1).
\end{enumerate}
\end{proofof}

\begin{proofof}{Proposition~\ref{p:effadm:constructionI}}
 See \cite[Proposition 4.3.14 and Lemma 4.5.12]{Sch:phd}.
\end{proofof}

\begin{proofof}{Proposition~\ref{p:effadm:constructionII}}
\begin{enumerate}
  \item See \cite[Propositions 4.3.13 and 4.2.5(4)]{Sch:phd}.
  \item See \cite[Propositions 4.5.1(3) and 4.3.6(2)]{Sch:phd}.
  \item See \cite[Propositions 4.4.1 and 4.3.6(2)]{Sch:phd}. 
  \item See \cite[Propositions 4.4.1 and 4.3.6(2)]{Sch:phd}.   
  \item See \cite[Propositions 4.4.5(2) and 4.3.6(2)]{Sch:phd}.     
  \item See \cite[Propositions 4.4.9(1),(3) and 4.3.6(2)]{Sch:phd}. 
  \item See \cite[Propositions 4.4.9(2) and 4.3.6(2)]{Sch:phd}.   
\end{enumerate}
\end{proofof}

\begin{proofof}{Example~\ref{ex:EffQCB}}
  This can be deduced from \cite[Proposition 4.3.16]{Sch:phd}.
  For (2) and (3) see also Proposition~\ref{p:delta^IS:properties}(4).
  For (5) see \cite[Lemma 4.3.6(3)]{Sch:phd}.
\end{proofof}

\begin{proofof}{Proposition~\ref{p:construction:effqcb}}
 See Propositions~\ref{p:effadm:constructionI} and~\ref{p:effadm:constructionII}.
\end{proofof}

\begin{proofof}{Theorem~\ref{th:effqcb:ccc}}
 See Propositions~\ref{p:effadm:constructionI} 
 and~\ref{p:effadm:constructionII}(2).
\end{proofof}

\begin{proofof}{Proposition~\ref{p:effqcb:compfunctions}}
  See \cite[Propositions 1, 3, 6, 11, 21, 40, 42]{Pau16}.
\end{proofof}

\begin{proofof}{Proposition~\ref{p:computable:effseqHaus}}
  See \cite[Proposition 14]{Pau16}.
\end{proofof}

\begin{proofof}{Theorem~\ref{th:QN:closure:properties}}
  See \cite[Theorem 6]{Sch09b}.
\end{proofof}

\begin{proofof}{Theorem~\ref{th:QN:uniform:Extension}}
  This follows from \cite[Theorem 13]{Sch09b} 
  and the fact that $\QN$ is cartesian closed 
  by Theorem~\ref{th:QN:closure:properties}.
\end{proofof}

\begin{proofof}{Proposition~\ref{p:delta^IR:properties}}
\begin{enumerate}
  \item See \cite[Proposition 4.3.5(3)]{Sch:phd}.
  \item See \cite[Proposition 4.3.5(2)]{Sch:phd}.  
  \item See \cite[Lemma 4.3.6(2)]{Sch:phd}.
  \item 
   This follows from the fact that the Cauchy representation $\varrho_\IR$
   for the reals is effectively admissible by Proposition~\ref{p:delta^IS:properties}(4).
  \item 
   This follows from Example~\ref{ex:cms}.
\end{enumerate}
\end{proofof}

\end{document}